\documentclass[10pt,journal,twocolumn]{IEEEtran} 

\usepackage{amsthm}
\usepackage{amsfonts}
\usepackage{amsmath}
\usepackage{amssymb}
\usepackage{graphicx}
\usepackage[colorlinks=truser, allcolors=blue]{hyperref}
\usepackage{bbm}
\usepackage{caption,subcaption}
\usepackage{verbatim}
\usepackage{cite,enumerate}
\usepackage{srcltx,cases}
\usepackage[mathscr]{eucal}
\usepackage{multirow}
\usepackage{comment}
\newtheorem{definition}{Definition}
\newtheorem{theorem}{Theorem}
\newtheorem{lemma}{Lemma}
\newtheorem{proposition}{Proposition}
\newtheorem{corollary}{Corollary}
\newtheorem{remark}{\textit{Remark}}

\usepackage[ruled,linesnumbered,vlined]{algorithm2e}
\usepackage[compact]{titlesec}
\usepackage{mathtools}

\usepackage{accents}

\usepackage{tikz}

\title{Joint Visibility Analysis of RIS in Non-Terrestrial Networks through Stochastic Geometry}

\author{Ashutosh Balakrishnan, Junse Lee, Fran\c{c}ois Baccelli 
\thanks{Ashutosh Balakrishnan is with Department of Computer Science and Networks, IMT-Télécom Paris, Institut Polytechnic de Paris (e-mail: balakrishnan@telecom-paris.fr).}\thanks{Junse Lee is with School of AI Convergence, Sungshin Women's University, Seoul, South Korea (e-mail: junselee@sungshin.ac.kr).}\thanks{Fran\c{c}ois Baccelli is with INRIA/ ENS and Télécom Paris (e-mail francois.baccelli@ens.fr).}}

\begin{document}
\maketitle

\begin{abstract}
Non-Terrestrial networks (NTNs)  
are a key theme in upcoming 6G communications, especially for ubiquitous coverage. Urban environments, comprising of high rise buildings often result in blocking the line of sight (LoS) path between the user equipment (UE) and the NTN base station (NTN-BS). In this paper we investigate the situation where reconfigurable intelligent surfaces (RIS) are deployed on the building roof-tops to ensure multi-hop connectivity between the UE and the NTN-BS. In such a scenario, it becomes crucial to statistically study the LoS visibility of the RIS from the UE as well as from the NTN-BS, hence termed as joint visibility. 
In this work, accounting for the dual stochasticity arising from the locations of the RIS deployed buildings  and the respective random building heights,  we statistically study the probability of  joint RIS visibility in a two-dimensional (2D) scenario considering a deterministic location of the NTN-BS. 
Further, we study the joint RIS visibility statistics conditional on the UE-NTN link being LoS or non-LoS. 
For the RISs deployed as a point point process (PPP) having exponentially distributed heights, the expected RISs jointly visible under the unconditional and conditional geometric settings are derived in closed form. Interestingly, in the 2D setting, the maximum expected RISs jointly visible, unconditionally, is twice the Basel number $(\pi^2\slash 6)$. 
The simulated results are analyzed over building density, average building height, the altitude and position of the NTN-BS. We also illustrate probability heatmaps, demonstrating the strongest chance to have a RIS used   conditioned on the system geometry. 
This study is expected to be useful in planning the deployment of RIS in urban areas, improving the signal and for assessing economic  aspects. 
\end{abstract}

\begin{IEEEkeywords}
RIS, Non-terrestrial network, blockages, line of sight, stochastic geometry, Point process, random shape.
\end{IEEEkeywords}

\IEEEpeerreviewmaketitle
\section{Introduction}

Non-Terrestrial Networks (NTNs) are expected to be a key theme of the upcoming sixth-generation (6G) communication systems. For, e.g., unmanned aerial vehicles (UAVs), high altitude platform stations (HAPS), and low earth orbit (LEO) satellites, form the different vertical layers of NTNs. These NTNs  are expected to play key role in providing ubiquitous coverage across the globe and aid the existing terrestrial communication infrastructure.

In urban areas, the presence of tall buildings often results in blocking the line of sight (LoS) communication link between the NTN  base station (BS) and the ground user equipment (UE). To this extent, Reconfigurable Intelligent Surfaces (RISs) can be deployed to extend the connectivity. RISs can be very useful as they can relay the signals intelligently, creating LoS multi-hop links between the aerial BS and UE. Below we motivate and position our work.

\subsection{Motivation}
A point (which may represent a UE, RIS, or an NTN-BS) is termed to be ``visible'' to another point, if the wireless link between the two points is not obstructed by any obstacle falling in between them. It is important to study, not only the visibility of a RIS (deployed on a building) from the UE, but also the visibility of the RIS from the aerial BS. Hence, the study of joint visibility of a point becomes very crucial, especially in urban scenarios, which are studied in this work.  

In this work, we study the global statistical properties of the UE-RIS-NTN system, accounting for the  stochasticity arising from the system geometry (i.e., randomness due to the locations and the height of the RIS deployed buildings).
Through the proposed framework, we derive expressions for  the mean number of jointly visible RISs and the 
intensity measure of the jointly visible RISs (unconditionally as well as conditional on the UE-NTN link being LoS or NLoS). These expressions are expected to  be instrumental in assessing the improvement 
of spectral efficiency brought by randomly distributed
and shared RISs in this urban NTN context. With rising deployment of NTN-BSs like UAVs, HAPS, and LEO satellites, the tools presented in this work are expected to be potent in planning the cellular infrastructure in urban cities, providing revenue gains to the mobile operator as well as signal quality improvement to the UE.

\subsection{Related Works} 

\subsubsection{Network Analysis using Stochastic Geometry}

Stochastic geometry has been widely utilized to analyze the behavior of wireless networks at the macroscopic level. By modeling network node locations as a realization of a random spatial point process, it provides a tractable analytical framework for characterizing network performance \cite{baccelli2009stochastic}. This approach has been applied to ad-hoc networks \cite{baccelli2006aloha,baccelli2009stochasticopp,lee2016spectral}, cellular networks \cite{andrews2011tractable,dhillon2012modeling}, and vehicular networks \cite{tong2016stochastic,yi2019modeling}.

Recently, this framework has been extended to the analysis of NTNs for evaluating network-level performance metrics such as visibility, connectivity, coverage, and handover rates \cite{chetlur2017downlink,banagar2020performance,okati2023stochastic,al2022next}. However, the impact of blockage effects caused by the terrestrial urban environment has not been fully characterized. Also, the performance gains of RIS deployments have been analyzed by modeling the RIS locations as point processes \cite{adrat2024performance, sun2025stochastic}. 

\subsubsection{Urban Blockage Modeling}

Traditionally, modeling the effect of blockage has relied on empirical approaches such as ITU-R propagation models \cite{ITUR_P1411}, which characterize the LoS probability as a function of the distance between two nodes. An alternative method for modeling this effect is the use of  random shape theory \cite{kendall1989survey}, which extends independent blockage models to analyze urban wireless networks \cite{bai2014analysis, blaszczyszyn2015wireless, ilow1998analytic}. To capture the spatial correlation of blockage, Poisson line process-based models have been proposed to analyze correlated shadowing in urban and in-building networks \cite{baccelli2015correlated, zhang2015indoor, lee20163}. However, these models are insufficient to capture the elevation-dependent distribution of blocking that governs the availability of NTN communication links. To address this, the LoS probability of outdoor links has been analyzed as a function of the elevation angle \cite{al2020probability, al2024line}. Furthermore, a multiplicative cascade blockage model has been introduced to characterize the blockage effect in urban environments, enabling a more refined analysis of signal obstruction by buildings \cite{baccelli2022user, liu2023macro, liu20263d}.

\subsubsection{Sky Visibility in NTN Networks}
In a prior work \cite{lee2024much}, a stochastic geometry framework was proposed to quantify the sky visibility for ground users in urban environments by modeling building locations and heights as a one-dimensional marked Poisson point process (PPP). This framework derived the blockage angle distribution along a given viewing direction and demonstrated how rooftop-mounted RISs can extend the user's sky visibility toward NTN nodes. Building on \cite{lee2024much}, the skyline process was proposed to extend the visibility analysis to a three-dimensional urban environment, characterizing the maximum blockage elevation angle as a continuous function of the azimuth angle \cite{lee2026skyline}. This framework analyzed the joint distribution of blockage angles across different viewing directions, along with second-order statistics including the autocorrelation function and power spectral density, to characterize the spatial correlation of blockage effects. The integration of RISs with NTNs has been widely studied to explore RIS's potential to extend coverage and improve link reliability in NTN environments \cite{amodu2024technical}. However, a stochastic geometry-based analysis of the visibility of RIS from both the ground user and the NTN node simultaneously remains unexplored, which is essential for characterizing RIS potential use in urban NTN deployments.

\subsection{Contributions}
The main technical contributions in this paper are summarized below.
\begin{itemize} 
\item \textbf{Joint RIS Visibility Analysis:} With the RIS locations modeled as a marked PPP, we derive closed-form expressions for the probability that a RIS is jointly visible from both the ground UE and the NTN-BS, under all possible 2D geometric configurations.

\item \textbf{Conditional Joint RIS Visibility:} Next, we derive the joint RIS visibility probability conditioned on the direct UE-NTN link being LoS or NLoS, providing analytical insights into the role of RIS under both LoS and NLoS propagation conditions.

\item \textbf{RIS Visibility Statistics from UE:} For exponentially distributed building heights, we derive the expected number of RISs visible to the UE and the probability that at least three RISs are visible from the UE.

\item \textbf{Expected Number of Jointly Visible RISs:} We then derive the expected number of jointly visible RISs, unconditionally as well as along with their LoS and NLoS conditional expectations, and analyze their limiting behaviors. Interestingly, the maximum expected RISs jointly visible, unconditionally, comes out to be twice the Basel number $(\pi^2 \slash 6)$.

\item \textbf{Results and Inference:} We present simulation based results illustrating the effect of building density, building height, the position and altitude of the NTN-BS on the expected RIS visibility. Probability heatmaps are presented, which illustrate the locations wherein the RISs are most useful, depending on the position and altitude of the NTN-BS. These results are expected to be beneficial in planning cellular infrastructure, potentially providing gains in spectral efficiency, and operator revenue.      

\end{itemize}

\subsection{Organization}
The paper is organized as follows. Section~\ref{sec:sec2} introduces the RIS assisted NTN system model. 
Section~\ref{sec:sec3} derives the joint RIS visibility probabilities for all geometric configurations. Section~\ref{sec:sec4} extends this analysis to the joint RIS  visibility conditional on the direct UE-NTN link being LoS or NLoS. Section~\ref{sec:sec5} presents the RIS availability analysis, with  
insights on the limiting behaviors. Section~\ref{sec:sec6} provides numerical results 
and Section~\ref{sec:sec7} concludes the paper.


\section{System Model}\label{sec:sec2}
\subsection{Network Model}

In this work, we consider an outdoor environment with a ground UE, an NTN-BS, and randomly distributed buildings of varying heights to analyze visibility between the  UE and the NTN-BS. For analytical tractability, we consider a two-dimensional (2D) setting in which the dimensions correspond to the UE's viewing direction and building height. Note that analyzing the network performance experienced by a UE requires a 3D model, consisting of two dimensions associated with the building location and a third dimension corresponding to the building's height. However, such models are more complex to analyze and are left for a future study.

We assume that the UE is located at the origin, i.e.,  $(0,0)$. Next, the NTN-BS is assumed to be located at a coordinate $(X_n, H_n)$, where $X_n \geq 0$ without loss of generality. In an outdoor environment, buildings act as obstacles. To characterize the UE’s LoS visibility, we model the buildings as a stationary marked point process, with locations $\Phi = \{x_i\}$ following a 1D homogeneous PPP of intensity $\lambda$ and associated marks $H = \{h_i\}$ representing the building heights. The marks are i.i.d., following a distribution $F(\cdot)$, and are independent of $\Phi$. Let $\mu$ denote the inverse of the mean building height. 

\subsection{RIS-Assisted Link Geometry}\label{subsec:geometry}

\begin{figure}
    \centering
    \includegraphics[width=0.48\textwidth]{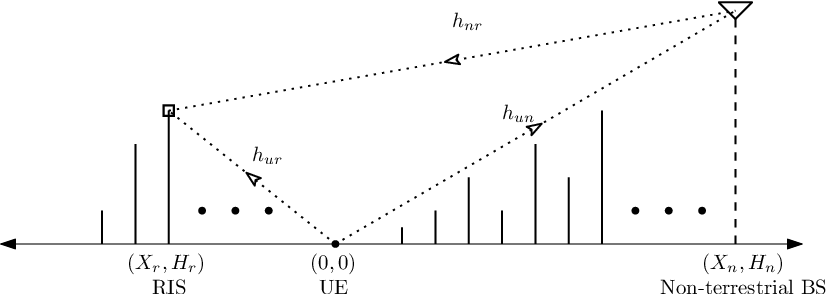}
    \caption{System model when the RIS is located on the left of the UE. Note that the RIS is visible jointly from both the UE as well as the NTN-BS.}
    \label{fig:system_LHS}
\end{figure}

\begin{figure*}[t]
    \centering
    \includegraphics[width=0.75\textwidth]{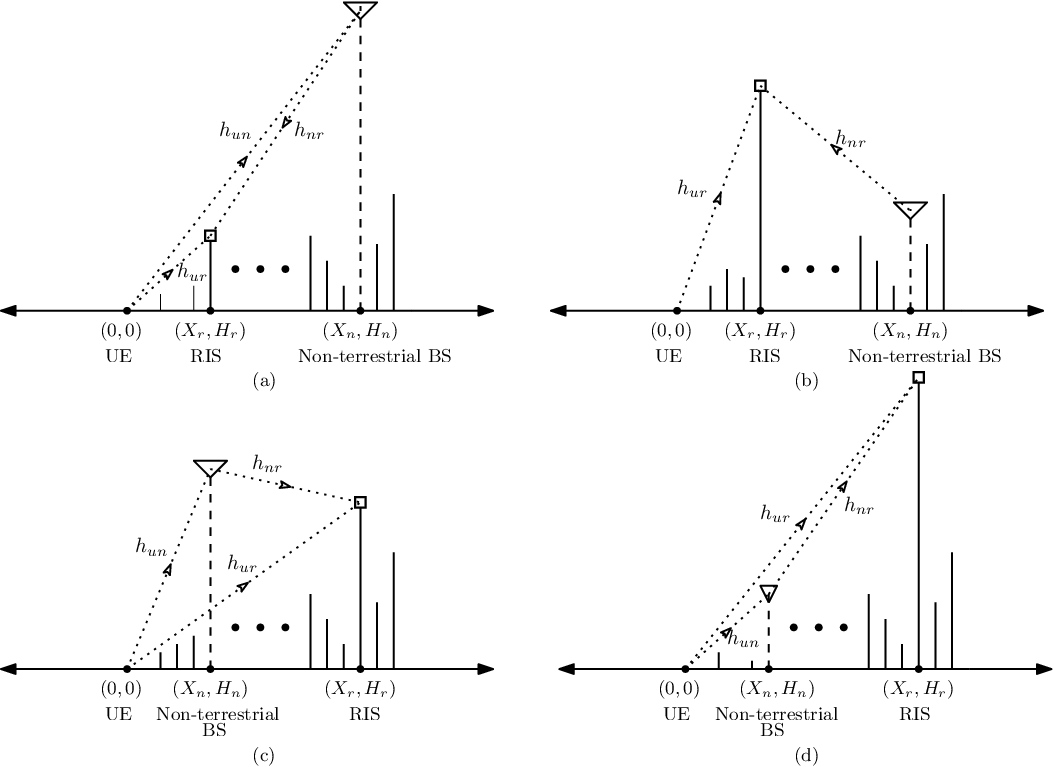}
    \caption{System models when the RIS is located on the RHS to the UE (a) case 1, (b) case 2, (c) case 3, (d) case 4. Note that the RIS requires to be visible jointly from the UE and the NTN-BS.}
    \label{fig:system_RHS}
\end{figure*}

To enhance the visibility between the UE and NTN-BS, it is assumed that a RIS is installed on the rooftop of each building. When the direct LoS path between the UE and a NTN node is blocked by a building, the rooftop RIS may serve as a relay, enabling a single-hop communication with NTN nodes. We consider two RIS operational modes: the transmissive mode, in which the signal from NTN-BS is forwarded in the same direction, and the reflective mode, in which the signal from the NTN-BS is reflected toward the opposite direction \cite{lee2024much}. Let the location of a RIS of interest be denoted as $(X_r, H_r)$, where $X_r \in \Phi$ and $H_r$ is independent of $X_r$.

Depending on the position of the RIS relative to the UE and the NTN-BS, there are two scenarios: (a) when the RIS is located on the left hand side (LHS) of the UE (i.e., $X_r < 0$, see Fig.~\ref{fig:system_LHS}))  and (b) when the RIS is located on  the right hand side (RHS) of the UE (i.e., $X_r \geq 0$, see Fig.~\ref{fig:system_RHS}). In the LHS scenario, the RIS always operates in reflective mode, reflecting the signal back toward the UE. In the RHS scenario, the operating mode depends on the position of the RIS relative to the NTN node. For instance, as inferred from Fig.~\ref{fig:system_RHS}, when $X_r \leq X_n$, the RIS operates in transmissive mode (Cases 1 and 2 in Fig.~\ref{fig:system_RHS}(a)-(b)). When $X_r > X_n$, the RIS operates in reflective mode (Cases 3 and 4 in Fig.~\ref{fig:system_RHS}(c)-(d)). 
Note that, in the RHS, the RIS deployment sub-cases can also be classified depending on the slope of the UE-NTN link and the slope of the  UE-RIS link. 
Based on the position of the RIS and  slope of the links, we classify the geometric regions as $D^{i}, i\in \{L, R_1, R_2, R_3, R_4\}$. These are summarized below:
\begin{align}
\text{Case LHS $(L)$:}& D^L {=} \{X_r < 0\}  \text{ (Fig.~\ref{fig:system_LHS})} \label{geometry-lhs-rhs1}
\\
\text{Case RHS-1\!\! $(R_1)$:}& D^{R_1} {=} \Bigg\{0 {\leq} X_r {<} X_n, \frac{H_r}{X_r} {\leq} \frac{H_n}{X_n}\Bigg\} \text{(Fig.~\ref{fig:system_RHS}(a))} \label{geometry-lhs-rhs2}
\\
\text{Case RHS-2\!\! $(R_2)$:}& D^{R_2} {=} \Bigg\{0  {\leq} X_r {<} X_n, \frac{H_r}{X_r} {>} \frac{H_n}{X_n}\Bigg\} \text{(Fig.~\ref{fig:system_RHS}(b))} \label{geometry-lhs-rhs3}
\\
\text{Case RHS-3\!\! $(R_3)$:}& D^{R_3} {=} \Bigg\{X_r {\geq} X_n, \frac{H_r}{X_r} {\leq} \frac{H_n}{X_n}\Bigg\}  \text{(Fig.~\ref{fig:system_RHS}(c))} \label{geometry-lhs-rhs4}
\\
\text{Case RHS-4\!\! $(R_4)$:}& D^{R_4} {=} \Bigg \{X_r {\geq} X_n, \frac{H_r}{X_r} {>} \frac{H_n}{X_n}\Bigg \} \text{(Fig.~\ref{fig:system_RHS}(d)).} \label{geometry-lhs-rhs5}
\end{align}

\subsection{LoS Visibility and Performance Metrics}
Two points are said to be mutually visible, or equivalently to satisfy the LoS condition, if no building in $\Phi$ obstructs the line segment connecting them. Throughout this paper, $N_u$ denotes the number of RISs visible from the UE, and $N_J$ the number of RISs 
jointly visible from both the UE and the NTN-BS. The probabilistic  characterization of these quantities, presented in  Section~\ref{sec:sec5}, is the main focus of this work.

\section{Joint RIS Visibility Analysis}\label{sec:sec3}

To characterize the joint visibility of a RIS from both the UE and the NTN-BS, we first derive the individual  link LoS probabilities. For an arbitrary RIS at $(x,h)$, let the LoS probabilities of the UE-RIS link  and the NTN-RIS link be $P_{ur}(x,h)$ and $P_{nr}(x,h)$,  respectively. 
These are subsequently used to derive the unconditional joint RIS visibility probability $P_j(x,h)$, defined as the probability that a RIS at $(x,h)$ is simultaneously visible from  both the UE (at the origin) and the NTN-BS (at $(X_n, H_n)$). Note that the RIS location $(x,h)$ is interchangeably used with $(X_r, H_r)$.

\subsection{UE -- RIS LoS Probability}\label{subsec:los_prob}

In this subsection, we derive the LoS probability of the UE-RIS link. Note that this derivation holds for any RIS location, i.e., for the RIS located on the LHS or the RHS to the UE.

\begin{lemma}\label{lem:pU}
    When the height of the buildings is distributed as $F(\cdot)$, the LoS probability between the UE at $(0,0)$ and an arbitrary RIS $(x,h)$ is given by
\begin{equation}\label{eq:user-arb}
   P_{ur}(x,h) = \exp\left( - \lambda |x| \int_0^1 \left[1 - F(hu) \right] du \right). 
\end{equation}
\end{lemma}

\begin{proof}
    See Appendix \ref{app:lemma1}.
\end{proof}

\begin{corollary}\label{cor:UR_LoS_exp}
    For exponentially distributed building heights, i.e., $F(h) = 1 - e^{-\mu h}$, the LoS probability between UE and a RIS deployed at $(x,h)$ is 
    \begin{equation}\label{UR-LoS}
        P_{ur}(x,h) =  \exp\left(  \frac{- \lambda |x|}{\mu h}[1 - e^{-\mu h}] \right).
    \end{equation}
    This follows by substituting $F(h) = 1 - e^{-\mu h}$ into \eqref{eq:user-arb}.
\end{corollary}

\begin{remark}
The following observations provide intuition on $P_{ur}(x,h)$ in 
Corollary~\ref{cor:UR_LoS_exp}:
\begin{enumerate}
    \item as $h \rightarrow \infty$, $P_{ur}(x,h) \rightarrow 1$. So, if the height of the building $(x,h)$ is very large, then it is very highly probable to be visible to the UE at $(0,0)$. 

    \item as $h \rightarrow 0$, we have
    \begin{align*}
        \lim_{h \rightarrow 0} P_{ur}(x,h) &= \lim_{h \rightarrow 0} \exp\left(  \frac{- \lambda |x|}{\mu h}[1 - e^{-\mu h}] \right) \\
        & \text{approximating using Taylor series} \notag \\ 
        &\approx \lim_{h \rightarrow 0} \exp\left( \frac{- \lambda |x| (\mu h)}{\mu h} \right)   \approx \exp(- \lambda |x|)\mbox{.}
     \end{align*}
     The last equation is the probability that no points (i.e., buildings) lie between $(0, x)$ for $x>0$, or $(x,0)$ for $x<0$. That is for visibility, as $h \rightarrow 0$, no other point/ building should lie in the LoS between $(0,0)$ ``UE'' and $(x,h)$ ``building of interest''.

    \item as $|x|$ increases, i.e., the building of interest is more away from the UE at $(0,0)$, then $P_{ur}(x,h)$ decreases with increasing $x$ (i.e., higher probability of potential blockers between the building of interest and the UE).

    \item as $\lambda$ increases, $P_{ur}(x,h)$ decreases. Intuitively, more buildings probabilistically decrease the visibility probability between $(0,0)$ and $(x,h)$. 

    \item as $\mu$ increases, $P_{ur}(x,h)$ increases. Recall that $\mu$ is the inverse of the average height of buildings. So if $\mu$ increases, the  average height of buildings decreases, which increases the probability of LoS link visibility. 
\end{enumerate}
\end{remark}

\subsection{NTN-BS -- RIS LoS Probability} 
We next derive $P_{nr}(x,h)$, the LoS link probability between a RIS deployed at $(x,h)$ and the NTN-BS at $(X_n, H_n)$. 
\begin{lemma}\label{lem:pN}
The LoS probability between an arbitrary RIS at $(x, h)$ and the NTN-BS at $(X_n, H_n)$ is 
\begin{equation}
    P_{nr}(x,h) {=} \exp\left(-\lambda |X_n {-} x| 
    \int_0^1 \left[1 {-} F\left(h {+} u(H_n {-} h)\right)\right] \mathrm{d}u\right).
\end{equation}
\end{lemma}

\begin{proof}
    See Appendix \ref{app:lemma2}.
\end{proof}

\begin{corollary}\label{cor:RS_LoS_exp}
For exponentially distributed building heights, i.e., $F(h) = 1 - e^{-\mu h}$,
the LoS probability between an arbitrary RIS at $(x, h)$ and the NTN-BS 
at $(X_n, H_n)$ is 
\begin{equation}\label{eq:RS_LoS_exp}
    P_{nr}(x,h) = \exp\left(-\lambda |X_n - x| e^{-\mu h} 
    \frac{1 - e^{-\mu(H_n - h)}}{\mu(H_n - h)}\right).
\end{equation}
This follows by substituting $F(h) = 1 - e^{-\mu h}$ into Lemma~\ref{lem:pN}.
\end{corollary}

\begin{remark}
The following observations provide intuition on $P_{nr}(x,h)$ in \eqref{eq:RS_LoS_exp}:
\begin{enumerate}
\item as $|H_n - h| \rightarrow \infty$, $P_{nr}(x,h) \rightarrow 1$. This is intuitive, as large vertical height difference between the RIS and the NTN-BS (regardless of $H_n {>} h$ or $h {>} H_n$), indicates that the NTN-RIS link is highly probable to be in LoS. 

\item as $|H_n - h| \rightarrow 0$, we have
\begin{align*}
    &\lim_{|H_n {-} h| \rightarrow 0} P_{nr}(x,h) 
    \\&{=} \lim_{|H_n {-} h| \rightarrow 0} \exp\left( -\lambda|X_n {-} x|e^{-\mu h} \frac{1 {-} e^{-\mu(H_n {-} h)}}{\mu(H_n {-} h)} \right) 
    \\
    & \text{approximating using Taylor series} \notag 
    \\ 
    &\approx \exp\left(-\lambda |X_n - x| e^{-\mu h}\right)\mbox{.}
 \end{align*}
 The last equation is the probability that no buildings of height $h$ or higher lie between the NTN-BS and RIS at the same altitude. That is, for visibility, as $|H_n - h| \rightarrow 0$, no other building taller than $h$ should lie in the LoS between $(x,h)$ and $(X_n, H_n)$.

\item as $|X_n {-} x|$ increases, i.e., the building of interest is further away horizontally, then $P_{nr}(x,h)$ decreases when increasing $|X_n {-} x|$ (i.e., there is a higher probability of more potential blockers between the RIS and NTN-BS).

\item as $\lambda$ increases, $P_{nr}(x,h)$ decreases. Intuitively, more buildings probabilistically decrease the visibility probability between the RIS and the NTN-BS. 

\item as $\mu$ increases, $P_{nr}(x,h)$ increases. If $\mu$ increases, the average building height decreases, increasing the probability of LoS link visibility between RIS and NTN-BS. 

\end{enumerate}
\end{remark}

In particular, substituting $(x,h) = (X_r, H_r)$ into Lemma~\ref{lem:pU} and Lemma~\ref{lem:pN} yields $P_{ur}(X_r, H_r)$ and $P_{nr}(X_r, H_r)$, the LoS probabilities of the UE-RIS and the NTN-RIS links, respectively. Similar to \eqref{UR-LoS}, for exponentially distributed heights of buildings, the LoS probability of the direct UE-NTN link is 
\begin{equation}\label{p_los}
        P_{un}(X_n,H_n) \triangleq \exp\left(  \frac{- \lambda X_n}{\mu H_n}[1 - e^{-\mu H_n}] \right).
\end{equation}

\subsection{Joint RIS Visibility}\label{subsec:joint_vis}

In this subsection, we derive the joint RIS visibility probability $(P_j(x,h))$ using the probability of UE-RIS link visibility $(P_{ur}(x,h))$ and the probability of NTN-RIS link visibility $(P_{nr}(x,h))$ from Section~\ref{subsec:los_prob}. The analysis is sub-divided into two scenarios based on the position of the RIS relative to the UE and the NTN-BS: (i) the RIS is located on the LHS to the UE with $X_r < 0$, and (ii) the RIS is located on the RHS to the UE with  $X_r \geq 0$.

We first define the heights of the possible blockages which may fall in the link/ line segments of interest, i.e., UE-RIS, RIS-NTN, and the UE-NTN links. These heights are functions of the slope of the line segments previously described, which are of interest in our analysis. 
For a RIS at $(X_r, H_r)$ and an NTN-BS  at $(X_n,H_n)$, we define
\begin{align}
    h_{ur}(x) &= \frac{H_r}{X_r}x, \label{h_ur} \\
    h_{nr}(x) &= \frac{H_n - H_r}{X_n - X_r}(x - X_r) + H_r, \label{h_nr} \\
    h_{un}(x) &= \frac{H_n}{X_n}x. \label{h_un}
\end{align}
Here, $h_{ur}(x)$, $h_{nr}(x)$, and $h_{un}(x)$ denote the heights of the UE-RIS, RIS-NTN, and the  UE-NTN  line segments at some location $x$, respectively.

\begin{definition}
    Using the probability generating functional (PGFL) of a homogeneous PPP, for a blocking height function 
    $h(\cdot): [a,b] \to \mathbb{R}^+$, 
    where building heights follow a distribution $F(\cdot)$, we define 
    the visibility function as
\begin{equation}\label{visibility-function}
    \mathcal{V}(a, b, h(\cdot)) \triangleq \exp\left(-\lambda \int_a^b 
    \left[1-F(h(x))\right] \mathrm{d}x\right).
\end{equation}
For the line segments given in \eqref{h_ur}, \eqref{h_nr}, and \eqref{h_un}, with $F(h) = 1 {-} e^{-\mu h}$, we have
\begin{equation}
    \begin{aligned}
        &\mathcal{V}(a, b, h_{ur}(x)) = \exp\left( \frac{\lambda X_r}{\mu H_r}\left[ e^{-\mu \frac{H_r}{X_r}b} - e^{-\mu \frac{H_r}{X_r}a} \right] \right), 
        \\
        &\mathcal{V}(a, b, h_{nr}(x)) {=} \exp\left( \frac{\lambda (X_n {-} X_r)}{(\mu (H_n {-} H_r)} e^{-\mu \left(\frac{H_rX_n {-} H_nX_r}{X_n {-} X_r}\right)} \right. \notag \\&  \hspace{40mm} \left. \left[ e^{-\mu \left(\frac{H_n {-} H_r}{X_n {-} X_r}\right)b } {-} e^{-\mu \left(\frac{H_n {-} H_r}{X_n {-} X_r}\right)a }  \right] \right),
        \\
        &\mathcal{V}(a, b, h_{un}(x)) = \exp\left( \frac{\lambda X_n}{\mu H_n}\left[ e^{-\mu \frac{H_n}{X_n}b} - e^{-\mu \frac{H_n}{X_n}a} \right] \right).
    \end{aligned}
\end{equation}
\end{definition}

Now that we have all the tools required to study the joint visibility of a RIS from the UE as well as from the NTN-BS, we derive the joint visibility probability for the cases \eqref{geometry-lhs-rhs1}-\eqref{geometry-lhs-rhs5} below.

\subsubsection{RIS located on the LHS to UE $(X_r < 0)$}

We first consider the case where the RIS is on the LHS, i.e., $X_r<0$. In this case, the RIS operates in reflective mode.  
\begin{lemma}[LHS, Reflective mode of RIS operation]\label{prop:LHS}
For a RIS located at $(X_r, H_r)$ with $X_r < 0$ (i.e., LHS to the UE), the probability of joint RIS visibility, from the UE  and NTN-BS, is
\begin{align}
    P_j^{L}(X_r, H_r) &= \mathcal{V}(X_r, 0, h_{ur}(\cdot))\times \mathcal{V}(0, X_n, h_{nr}(\cdot)).
\end{align}
\end{lemma}
\begin{proof}
Through first principles, we define and compute the joint visibility of a RIS when it is located on the LHS as
    \begin{align*}
    &P_j^{L}(X_r, H_r) {\triangleq} \mathbb{P}\left(h_{ur} \text{LoS} \in [X_r, 0], h_{nr} \text{LoS} \in [X_r, X_n]\right)
    \\
    &{=} \mathbb{P}\left(\min\{h_{ur},  h_{nr}\} \text{LoS}  \in [X_r, 0]\right) \mathbb{P}\left( h_{nr} \text{LoS} \in [0, X_n]\right)
    \\
    &{=} \mathbb{P}\left(h_{ur}  \text{LoS} \in [X_r, 0]\right) \mathbb{P}\left( h_{nr} \text{LoS} \in [0, X_n]\right).
\end{align*}

The last  analysis shows that the RIS at $(X_r, H_r)$ is jointly visible from both UE and the NTN-BS if and only if no building in $[X_r, 0]$  and $[0, X_n]$ obstruct the UE-RIS and the NTN-RIS  links, respectively.  
Hence, using \eqref{visibility-function}, we get
\begin{equation*}
    P_j^{L}(X_r, H_r) = \mathcal{V}(X_r, 0, h_{ur}(\cdot))\times \mathcal{V}(0, X_n, h_{nr}(\cdot)).
\end{equation*}

\end{proof}

\subsubsection{RIS located on the RHS to UE $(X_r \geq 0)$}
We now consider the sub-cases for the RIS located on the  RHS of the UE, \eqref{geometry-lhs-rhs2} - \eqref{geometry-lhs-rhs5}. These sub-cases are classified on the basis of the position of the RIS relative to the NTN-BS and the slope of the UE-NTN and NTN-RIS links. 
Below, we derive the probability of joint RIS visibility for each of these sub-case on the RHS. Note that, we denote each RHS sub-case, i.e., case 1 in  \eqref{geometry-lhs-rhs2}, case 2 in \eqref{geometry-lhs-rhs3}, case 3 in \eqref{geometry-lhs-rhs4}, and case 4 in  \eqref{geometry-lhs-rhs5}, as $R_i \ \forall \ i \in \{1, 2, 3, 4\}$ and its corresponding probability of joint visibility as $P_j^{R_i}$.  Note that, when the RIS is in  RHS to the UE, the sub-cases 1 and 2 correspond to transmissive modes of RIS operation, while sub-cases 3 and 4 correspond to reflective modes of the RIS operation.

\begin{lemma}[RHS, Transmissive mode of RIS operation]\label{prop:RHS_transmissive}
For a RIS located on the RHS to the UE, operating in transmissive mode (i.e., case 1 in \eqref{geometry-lhs-rhs2} and case 2 in \eqref{geometry-lhs-rhs3}), the joint RIS visibility probabilities are

\textit{Case 1} ($0\leq X_r < X_n, \frac{H_r}{X_r} \leq \frac{H_n}{X_n}$):
\begin{align}\label{result-unc-R1}
    P_j^{R_1}(X_r, H_r) {=} \mathcal{V}(0, X_r, h_{ur}(\cdot)) \mathcal{V}(X_r, X_n, h_{nr}(\cdot)),  
\end{align}

\textit{and, Case 2} ($0\leq X_r < X_n, \frac{H_r}{X_r} > \frac{H_n}{X_n}$):
\begin{align}
    P_j^{R_2}(X_r, H_r) {=} \mathcal{V}(0, X_r, h_{ur}(\cdot)) \mathcal{V}(X_r, X_n, h_{nr}(\cdot)). 
\end{align}
\end{lemma}

\begin{proof}
    For the settings where the RIS is located as in \eqref{geometry-lhs-rhs2} and \eqref{geometry-lhs-rhs3}, the joint RIS visibility is defined and computed as below. 
\begin{align}
    &P_j^{R_1}(X_r, H_r) {=} P_j^{R_2}(X_r, H_r)  \notag \\
    &= \mathbb{P}(h_{ur} \text{LoS} \in [0, X_r], h_{nr} \text{LoS} \in [X_r, X_n]) \notag
    \\
    &= \mathbb{P}\left(h_{ur} \ \text{LoS}  \in [0, X_r]\right)  \mathbb{P}\left( h_{nr} \ \text{LoS} \in [X_r, X_n]\right) \label{rhs_case1_1}
    \\
    &= \mathcal{V}(0, X_r, h_{ur}(\cdot))\, \mathcal{V}(X_r, X_n, h_{nr}(\cdot)). \label{rhs_case1_2}
\end{align}
Hence, a RIS located at $(X_r, H_r)$ is jointly visible from the UE and the NTN-BS if and only if no buildings in $[0, X_r]$  and $[X_r, X_n]$ obstruct the UE-RIS and NTN-RIS link, respectively. Since $[0, X_r]$ and $[X_r, X_n]$ are disjoint, the two events are independent, resulting in \eqref{rhs_case1_1}. Further, using \eqref{visibility-function}, we get \eqref{rhs_case1_2}.

\end{proof}

\begin{lemma}[RHS, Reflective mode of RIS operation]\label{prop:RHS_reflective}
For a RIS located on the RHS to the UE, operating in reflective mode (i.e., case 3 in \eqref{geometry-lhs-rhs4} and case 4 in \eqref{geometry-lhs-rhs5}), the joint RIS  visibility probabilities are

\textit{Case 3} ($X_r \geq X_n, \frac{H_r}{X_r} \leq \frac{H_n}{X_n}$):
\begin{align}
    P_j^{R_3}(X_r, H_r) &= \mathcal{V}(0, X_r, h_{ur}(\cdot)) 
\end{align}

\textit{and, Case 4} ($X_r \geq X_n, \frac{H_r}{X_r} > \frac{H_n}{X_n}$):
\begin{align}
    P_j^{R_4}(X_r, H_r) &= \mathcal{V}(0, X_n, h_{ur}(\cdot))\, \mathcal{V}(X_n, X_r, h_{nr}(\cdot)). 
\end{align}
\end{lemma}

\begin{proof}
For the geometric settings where the RIS is located as in \eqref{geometry-lhs-rhs4}, the joint RIS visibility probability is defined as computed as provided below:
\begin{align}
        P_j^{R_3}(X_r, H_r) &{=}  \mathbb{P}(h_{ur}  \text{LoS} \in [0, X_r],  h_{nr} \text{LoS} \in [X_n, X_r]) \notag
    \\
    &= \mathbb{P}\left(h_{ur} \text{LoS} \in [0, X_r]\right)
    \label{rhs_case3_1}
    \\
    &= \mathcal{V}(0, X_r, h_{ur}(\cdot)). \label{rhs_case3_2}
    \end{align}

The simplification in \eqref{rhs_case3_1}  
can be inferred from the geometry (Fig. \ref{fig:system_RHS}(c)): since $h_{nr}(x) \geq h_{ur}(x) \ \forall \ x \in [X_n, X_r]$, the NTN-RIS link is LoS whenever the UE-RIS link is LoS on $[X_n, X_r]$. Therefore, the joint visibility event reduces to the UE-RIS link being unobstructed over the entire interval of $[0, X_r]$. Hence, using  \eqref{visibility-function}, we get \eqref{rhs_case3_2}.

Next, for the subcase 4 when the RIS is on the RHS to the UE given in \eqref{geometry-lhs-rhs5}, we have 
    \begin{align}
       &P_j^{R_4}(X_r, H_r) =  \mathbb{P}(h_{ur}  \text{LoS} \in [0, X_r],  h_{nr} \text{LoS} \in [X_n, X_r]) \notag
    \\
    &= \mathbb{P}(h_{ur} \text{LoS} \in [0, X_n])  \mathbb{P}(h_{nr} \text{LoS} \in [X_n, X_r])  \label{rhs_case4_1}
    \\
    &= \mathcal{V}(0, X_n, h_{ur}(\cdot))\, \mathcal{V}(X_n, X_r, h_{nr}(\cdot)). \label{rhs_case4_2}
    \end{align}
The difference in subcase 4, as compared to subcacse 3, comes from the geometry (Fig. \ref{fig:system_RHS}(d))  from which   
it can be inferred that $h_{ur}(x) \geq h_{un}(x) \ \forall \ x \in [0, X_n]$ and hence the line $h_{nr}(x)$ governs the blockage on $[X_n, X_r]$. Therefore, the joint RIS visibility event requires no building to obstruct $h_{ur}(x)$ on $[0, X_n]$ and  $h_{nr}(x)$ on $[X_n, X_r]$, respectively. Since these two intervals are disjoint, the two events are independent which gives  \eqref{rhs_case4_1}. Finally using \eqref{visibility-function}, we get  \eqref{rhs_case4_2}, which proves the result.

\end{proof}



\section{Conditional Joint RIS Visibility }\label{sec:sec4}

The analysis in the previous section pertained to the derivation of unconditional probabilities of the joint RIS visibility. In this section, we study the joint RIS visibility from the UE and the NTN-BS, conditional on the existence of the direct UE-NTN LoS or the NLoS.

\subsection{Conditional on the UE-NTN LoS}
 In this subsection, we present the statistics of the  joint RIS visibility probability conditional on the existence of UE-NTN LoS link. 

\begin{lemma}[LHS, LoS]\label{lemma-lhs-los}
    Conditional on the UE-NTN link LoS, for a RIS located on the LHS as in \eqref{geometry-lhs-rhs1}, the joint RIS visibility probability from the UE and the NTN-BS is
    \begin{align}
     P_{j|\mathrm{LoS}}^{L}(X_r, H_r) &= \mathcal{V}(X_r, 0, h_{ur}(\cdot)), 
 \end{align}
 where $\mathcal{V}(\cdot)$ represents the visibility function defined in \eqref{visibility-function}. 
\end{lemma}

\begin{proof}
    See Appendix \ref{app:lemma6}.
\end{proof}

\begin{lemma}[RHS, Transmissive-LoS]\label{prop:RHS_transmissive-los}
Conditional on the UE-NTN link LoS, for a RIS located on the RHS operating in transmissive mode (i.e., case 1 in \eqref{geometry-lhs-rhs2} and case 2 in \eqref{geometry-lhs-rhs3}), the joint visibility probabilities are 

\textit{Case 1} ($0\leq X_r < X_n, \frac{H_r}{X_r} \leq \frac{H_n}{X_n}$):
\begin{align}
    P_{j|\mathrm{LoS}}^{R_1}(X_r, H_r) &= \frac{P_j^{R_1}(X_r, H_r)}{\mathcal{V}(0, X_n, h_{un}(\cdot))} 
\end{align}

\textit{and, Case 2} ($0\leq X_r < X_n, \frac{H_r}{X_r} > \frac{H_n}{X_n}$):
\begin{align}
    P_{j|\mathrm{LoS}}^{R_2}(X_r, H_r) &= 0. 
\end{align}
\end{lemma}

\begin{proof}
    See Appendix \ref{app:lemma7}.
\end{proof}

\begin{lemma}[RHS, Reflective - LoS]\label{prop:RHS_reflective-los}
Conditional on the UE-NTN link LoS, for a RIS located on the RHS operating in reflective mode (i.e., case 3 in \eqref{geometry-lhs-rhs4} and case 4 in \eqref{geometry-lhs-rhs5}), the joint visibility probabilities are 

\textit{Case 3} ($X_r \geq X_n, \frac{H_r}{X_r} \leq \frac{H_n}{X_n}$):
\begin{align}
    P_{j|\mathrm{LoS}}^{R_3}(X_r, H_r) &= \frac{P_j^{R_3}(X_r, H_r)}{\mathcal{V}(0, X_n, h_{un}(\cdot))} 
\end{align}
and, \textit{Case 4} ($X_r \geq X_n, \frac{H_r}{X_r} > \frac{H_n}{X_n}$):
\begin{align}
    P_{j|\mathrm{LoS}}^{R_4}(X_r, H_r) &= \mathcal{V}(X_n, X_r, h_{nr}(\cdot)). 
\end{align}
\end{lemma}

\begin{proof}
    See Appendix \ref{app:lemma8}.
\end{proof}

\subsection{Conditional on the UE-NTN NLoS}
In this subsection, we study the joint RIS visibility probability conditional on the UE-NTN NLoS link.

\begin{lemma}[LHS, NLoS]\label{lemma-lhs-nlos}
    Conditional on the UE-NTN link NLoS, the probability of joint visibility of a RIS from the UE as well as from the NTN-BS is
\begin{align}
       P_{j|\mathrm{NLoS}}^{L}(X_r, H_r) &= \notag
    \\
&\hspace{-2cm}\frac{P_{j|\mathrm{LoS}}^{L}(X_r, H_r) \left[\mathcal{V}(0, X_n, h_{nr}(\cdot)) - \mathcal{V}(0, X_n, h_{un}(\cdot))\right]}{1 - \mathcal{V}(0, X_n, h_{un}(\cdot))}.
\end{align}
\end{lemma}

\begin{proof}
    See Appendix \ref{app:lemma9}.
\end{proof}

\begin{lemma}[RHS, Transmissive]\label{prop:RHS_transmissive-nlos}
Conditional on the UE-NTN link NLoS, for a RIS located on the RHS operating in transmissive mode (i.e., case 1 in \eqref{geometry-lhs-rhs2} and case 2 in \eqref{geometry-lhs-rhs3}), the joint visibility probabilities are given as follows.

\textit{Case 1} ($0\leq X_r < X_n, \frac{H_r}{X_r} \leq \frac{H_n}{X_n}$):
\begin{align}
    P_{j|\mathrm{NLoS}}^{R_1}(X_r, H_r) &= 0
\end{align}

\textit{and, Case 2} ($0\leq X_r < X_n, \frac{H_r}{X_r} > \frac{H_n}{X_n}$):
\begin{align}
   P_{j|\mathrm{NLoS}}^{R_2}(X_r, H_r) &= \frac{P_j^{R_2}(X_r, H_r)}{1 - \mathcal{V}(0, X_n, h_{un}(\cdot))}.
\end{align}
\end{lemma}

\begin{proof}
    See Appendix \ref{app:lemma10}.
\end{proof}

\begin{lemma}[RHS, Reflective-NLoS]\label{prop:RHS_reflective-nlos}
Conditional on the UE-NTN link NLoS, for a RIS located on the RHS operating in reflective mode (i.e., case 3 in \eqref{geometry-lhs-rhs4} and case 4 in \eqref{geometry-lhs-rhs5}), the joint visibility probabilities are given as follows:

\textit{Case 3} ($X_r \geq X_n, \frac{H_r}{X_r} \leq \frac{H_n}{X_n}$):
\begin{align}
    P_{j|\mathrm{NLoS}}^{R_3}(X_r, H_r) &= 0,
\end{align}

\textit{Case 4} ($X_r \geq X_n, \frac{H_r}{X_r} > \frac{H_n}{X_n}$):
\begin{align}
    P_{j|\mathrm{NLoS}}^{R_4}(X_r, H_r) &= \notag \\
    &\hspace{-2cm}  \frac{\left(\mathcal{V}(0, X_n, h_{ur}(\cdot)) - \mathcal{V}(0, X_n, h_{un}(\cdot))\right)\mathcal{V}(X_n, X_r, h_{nr}(\cdot))}{1 - \mathcal{V}(0, X_n, h_{un}(\cdot))}.
\end{align}
\end{lemma}

\begin{proof}
    See Appendix \ref{app:lemma11}.
\end{proof}


\section{RIS Availability Analysis}\label{sec:sec5}
\subsection{Expected Number of Visible and Jointly Visible RISs}\label{subsec:expected_vis}

In this subsection, we derive the expected number of RISs visible from the UE, $\mathbb{E}[N_u]$, and the expected number of jointly visible RISs (from the UE and NTN-BS), $\mathbb{E}[N_J]$.

\begin{proposition}\label{prop:E_Nvis}
The expected number of RISs visible to the UE is given by 
\begin{equation}
    \mathbb{E}[N_u] = \int_{-\infty}^{\infty}\int_0^{\infty} \lambda f(h) P_{ur}(x,h)\mathrm{d}h\mathrm{d}x,
\end{equation}
where $P_{ur}(x,h)$ is given in Lemma~\ref{lem:pU}.
\end{proposition}
\begin{proof}
The result follows directly from Campbell's theorem for marked PPPs \cite{baccelli2009stochastic}.
\end{proof}

\begin{theorem}\label{theo:E_Nvis_exp}
For exponentially distributed building heights, i.e., $F(h) = 1 - e^{-\mu h}$, the expected number of RISs visible to the UE is 
\begin{equation}
    \mathbb{E}[N_u] = \frac{\pi^2}{3}.
\end{equation}
Note that this is also the maximum value of the expected number of RISs jointly visible (unconditionally) from the UE and NTN-BS. 
\end{theorem}
\begin{proof}
Substituting $F(h) = 1 - e^{-\mu h}$ and $P_{ur}(x,h)$ from Corollary~\ref{cor:UR_LoS_exp} into Proposition~\ref{prop:E_Nvis}, we get
\begin{align}
    \mathbb{E}[N_u] 
    &= \int_{-\infty}^{\infty}\int_0^{\infty}\lambda \mu e^{-\mu h} \exp\left(-\frac{\lambda |x|}{\mu h}
    [1-e^{-\mu h}]\right) \mathrm{d}h \mathrm{d}x\nonumber \\
    &= 2\int_0^{\infty} \frac{\mu^2 h e^{-\mu h}}{1-e^{-\mu h}}\mathrm{d}h.
\end{align}
Substituting $r = \mu h$,
\begin{align}
    \mathbb{E}[N_u] 
    &= 2\int_0^{\infty} \frac{r e^{-r}}{1-e^{-r}}\mathrm{d}r = 2\sum_{n=1}^{\infty} \frac{1}{n^2} = \frac{\pi^2}{3},
\end{align}
where the last equality follows from the Basel problem \cite{ivan2008simple}.
\end{proof}

\begin{theorem}
For exponentially distributed building heights, i.e., $F(h) = 1-e^{-\mu h}$, the probability that at least three RISs are visible from the UE is given by
\begin{equation}
    \mathbb{P}[N_u\geq 3] = 1-\left(\int_0^\infty \frac{te^{-t}}{t + e^{-t}} \mathrm{d}t\right)^2.
\end{equation}
\end{theorem}

\begin{proof}
    The RISs on the nearest buildings on each side are always visible. So, $\mathbb{P}[N_u < 2] = 0$. Also, let $N_u^+$ denote the number of RISs visible from the UE among those located on the RHS (i.e., $x > 0$). Then,
    \begin{equation}
        \mathbb{P}[N_u \geq 3] = 1-\left(\mathbb{P}[N_u^+ = 1]\right)^2.
    \end{equation}
    Conditioning on the nearest building at $(r,h)$ with $r>0$, the probability that all subsequent buildings are blocked is 
    \begin{equation}
    \exp\left(-\lambda \int_r^\infty \left[1 - F\left(\frac{h}{r}s\right)\right] 
    \mathrm{d}s\right) = \exp\left(-\frac{\lambda r}{\mu h} e^{-\mu h}\right).
    \end{equation}
    Unconditioning over $(r, h)$ and integrating over $r$,
\begin{align}
    &\mathbb{P}[N_u^+ = 1] \notag \\
    &= \int_0^\infty \int_0^\infty \lambda e^{-\lambda r} \mu e^{-\mu h} 
    \exp\left(-\frac{\lambda r}{\mu h} e^{-\mu h}\right) \mathrm{d}r\, \mathrm{d}h 
    \nonumber \\
    &= \int_0^\infty   \frac{h\mu^2}{1 + h\mu e^{h\mu}} 
    \mathrm{d}h \nonumber \\
    &= \int_0^\infty \frac{te^{-t}}{t + e^{-t}} \mathrm{d}t,
\end{align}
where the last equality follows from the substitution $t=\mu h$.
\end{proof}

\begin{theorem}\label{theo:E_NJ}
The expected number of jointly visible RISs is given by
\begin{align}
    \mathbb{E}[N_J] {=} \sum_{i \in \{L, R_1, R_2, R_3, R_4\}} 
    \iint_{\mathcal{D}^i} \lambda f(h) P_j^{i}(x,h)\mathrm{d}h\mathrm{d}x.
\end{align}
Here, $D^i$ and $P_j^i(x,h)$ denote the geometric region and its joint RIS visibility probability (Lemma~\ref{prop:LHS}, \ref{prop:RHS_transmissive}, 
and~\ref{prop:RHS_reflective}), respectively, for $i \in \{L, R_1, R_2, R_3, R_4\}$ in \eqref{geometry-lhs-rhs1}-\eqref{geometry-lhs-rhs5}.

\end{theorem}

\begin{proof}
By Campbell's theorem for marked PPPs \cite{baccelli2009stochastic},
\begin{align}
    \mathbb{E}[N_J] 
    &= \int_{-\infty}^{\infty}\int_0^{\infty} 
    \lambda f(h) P_{j}^{i}(x,h) \mathrm{d}h \mathrm{d}x,
\end{align}
where $i \in \{L, R_1, R_2, R_3, R_4\}$. The result follows by decomposing the domain of integration into the five geometric regions and substituting the joint visibility probabilities $P_j^{i}(x,h)$ from Lemma ~\ref{prop:LHS} - Lemma\ref{prop:RHS_reflective}, respectively.

\end{proof}

\begin{corollary}
As the slope of UE-NTN link tends to infinity, i.e., $H_n/X_n \to \infty$, the expected number of jointly visible RISs converges to the expected number of visible RISs, i.e., $\mathbb{E}[N_J] \to \mathbb{E}[N_\mathrm{vis}] 
= \pi^2/3$. Conversely, as $H_n/X_n \to 0$, $\mathbb{E}[N_J] \to 0$.
\end{corollary}
\begin{proof}
As $H_n/X_n \to \infty$, $h_{nr}(x) \to \infty$, so $\mathcal{V}(\cdot,\cdot,h_{nr}(\cdot)) \to 1$ and $P_j^{i}(x,h) \to P_{ur}(|x|,h)$ for all regions, i.e., $i \in \{L, R_1, R_2, R_3, R_4\}$ in \eqref{geometry-lhs-rhs1}--\eqref{geometry-lhs-rhs5}. Hence, from Theorem 1, $\mathbb{E}[N_J] \to \mathbb{E}[N_\mathrm{vis}] = \pi^2/3$. Conversely, as $H_n/X_n \to 0$, the region  $R_3$ in \eqref{geometry-lhs-rhs4} vanishes. For the remaining regions, i.e., in $i \in \{L, R_1, R_2, R_4\}$,  $\mathcal{V}(\cdot,\cdot,h_{nr}(\cdot)) \to 0$ yielding $P_j^{i}(x,h) \to 0$, resulting in $\mathbb{E}[N_J] \to 0$.
\end{proof}

\begin{theorem}\label{theo:E_NJ_cond}
    The expected numbers of jointly visible RISs, conditional on the UE-NTN LoS and UE-NTN NLoS links are given by
    \begin{align}
    &\mathbb{E}[N_{J|\mathrm{LoS}}]\nonumber\\ &= \sum_{i \in \{L, R_1, R_2, R_3, R_4\}} 
    \iint_{\mathcal{D}^i} \lambda f(h) P_{j|\mathrm{LoS}}^{i}(x,h)\,\mathrm{d}h\,\mathrm{d}x, \\
    &\mathbb{E}[N_{J|\mathrm{NLoS}}]\nonumber\\ &= \sum_{i \in \{L, R_1, R_2, R_3, R_4\}} 
    \iint_{\mathcal{D}^i} \lambda f(h) P_{j|\mathrm{NLoS}}^{i}(x,h)\,\mathrm{d}h\,\mathrm{d}x.
\end{align}

Here, $P_{j|k}^{i}(x,h)$ corresponds to the conditional joint visibility probability (see Lemma \ref{lemma-lhs-los} -- Lemma  \ref{prop:RHS_reflective-nlos} corresponding to $D^i$ in \eqref{geometry-lhs-rhs1}--\eqref{geometry-lhs-rhs5}),  with $i \in \{L, \ldots R_4\}$ and $k \in \{\text{LoS}, \text{NLoS}\}$. 
\end{theorem}

\begin{proof}
The result directly follows from Campbell's theorem for marked PPPs \cite{baccelli2009stochastic}, by decomposing the domain of integration into the five geometric regions and substituting the conditional joint visibility probabilities $P_{j|\mathrm{LoS}}(x,h)$ and $P_{j|\mathrm{NLoS}}(x,h)$ from Lemma~\ref{lemma-lhs-los}--Lemma~\ref{prop:RHS_reflective-nlos}, respectively.
\end{proof}

\begin{remark}
The results of Theorem~\ref{theo:E_NJ} and Theorem~\ref{theo:E_NJ_cond} 
satisfy the law of total expectation, i.e.,
\begin{equation}\label{inference-expected-joint-visibility}
\begin{aligned}
    \mathbb{E}[N_J] {=}& \mathbb{E}[N_{J|\mathrm{LoS}}] P_{un}(X_n, H_n) 
    \\
    &{+} \mathbb{E}[N_{J|\mathrm{NLoS}}] (1 {-} P_{un}(X_n, H_n)), 
\end{aligned}
\end{equation}
where $P_{un}(X_n, H_n)$ is given in \eqref{p_los} and the proof follows directly from Lemma~\ref{lemma-lhs-los}--Lemma~\ref{prop:RHS_reflective-nlos}.
\end{remark}

\begin{corollary}
As $H_n/X_n \to \infty$, the expected number of jointly visible RISs conditional on LoS converges to the expected number of visible RISs, i.e., $\mathbb{E}[N_{J|\mathrm{LoS}}] \to \mathbb{E}[N_\mathrm{vis}] = \pi^2/3$.
\end{corollary}
\begin{proof}
As $H_n/X_n \to \infty$, $P_{un}(X_n, H_n) \to 1$ (refer to \eqref{p_los}), and hence the regions corresponding to regions $R_2$ in \eqref{geometry-lhs-rhs3} and $R_4$ in \eqref{geometry-lhs-rhs5} vanish. For the remaining cases, i.e., $i \in \{L, R_1, R_3\}$, $P^i_{j|\mathrm{LoS}}(x,h) \to P_{ur}(|x|,h)$, and so by Campbell's theorem, $\mathbb{E}[N_{J|\mathrm{LoS}}] \to \mathbb{E}[N_\mathrm{vis}] = \pi^2/3$.
\end{proof}

\section{Numerical Results}\label{sec:sec6}

The simulation parameters used in this section are: 
$\lambda = \{0.001, 0.007, 0.012\},$ corresponding to rural, urban, and dense urban settings, respectively  \cite{lee2024much,3gpp2020study}. The inverse of the average height of building is $\mu = 0.02$, abscissa of the NTN-BS $X_n = 200$, and altitudes of the NTN-BS $H_n= \{100,  10000\}$m. These altitudes correspond to a low altitude UAV and a high altitude NTN,  respectively.

\subsection{Parameter Sensitivity Analysis}

\begin{figure}
    \centering
    \includegraphics[width=0.45\textwidth]{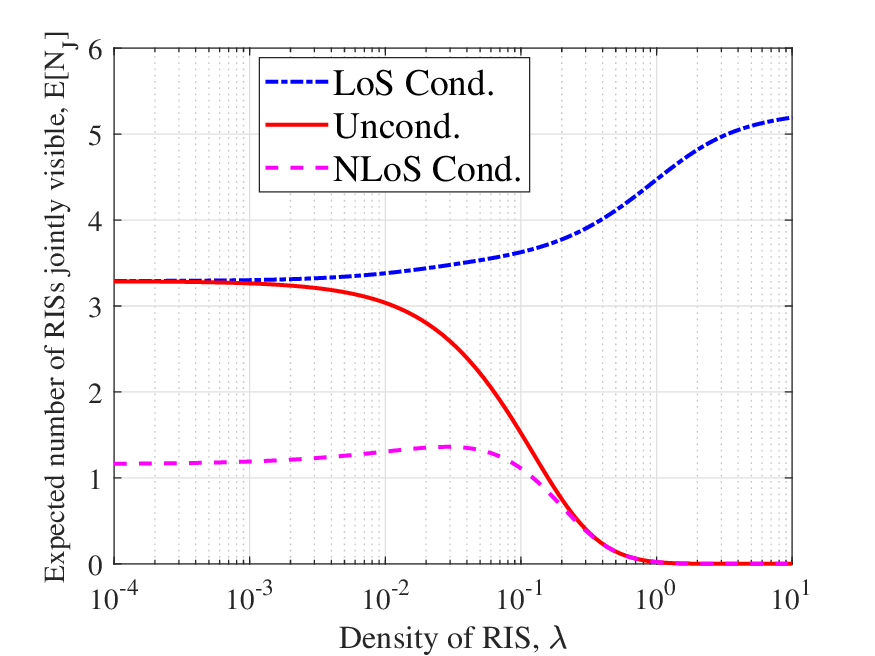}
    \caption{Variation of expected number of RISs jointly visible with the density of buildings, $\lambda$ in a PPP setting.}
    \label{fig:3}
\end{figure}

\begin{figure*}
    \centering
    \subfloat[]{
    \includegraphics[width=0.3\linewidth]{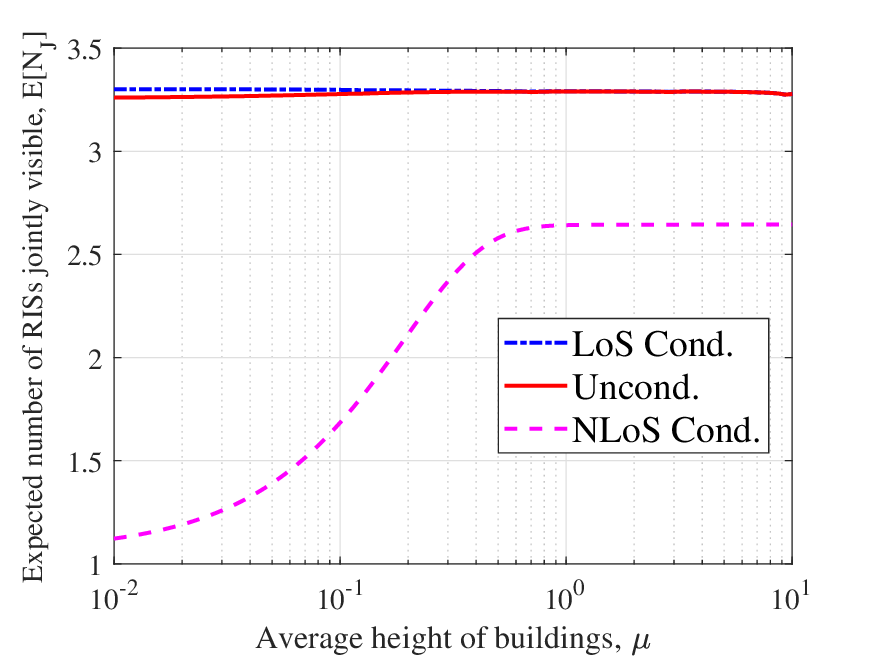}}
    \subfloat[]{
    \includegraphics[width=0.3\linewidth]{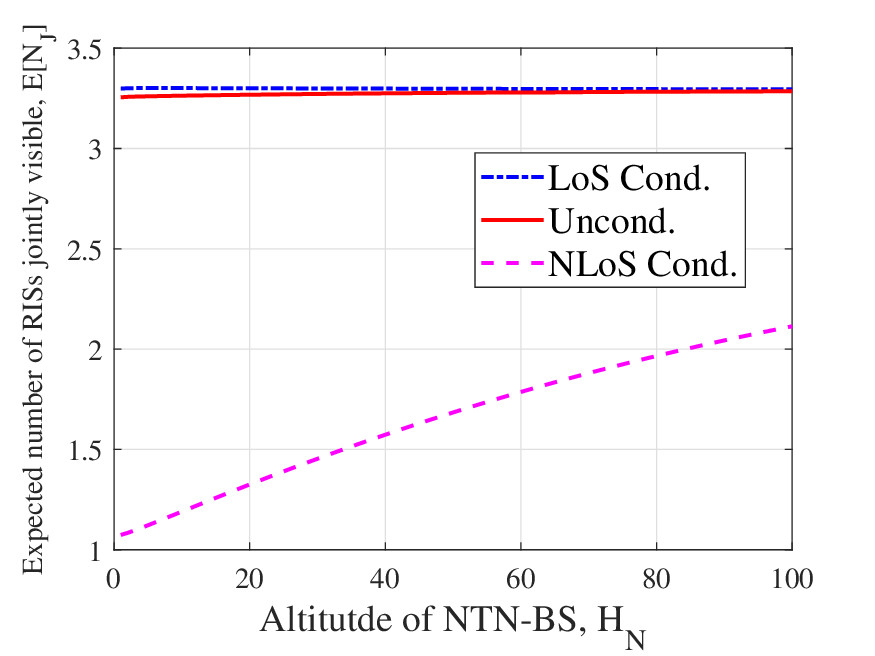}}
    \subfloat[]{
    \includegraphics[width=0.3\linewidth]{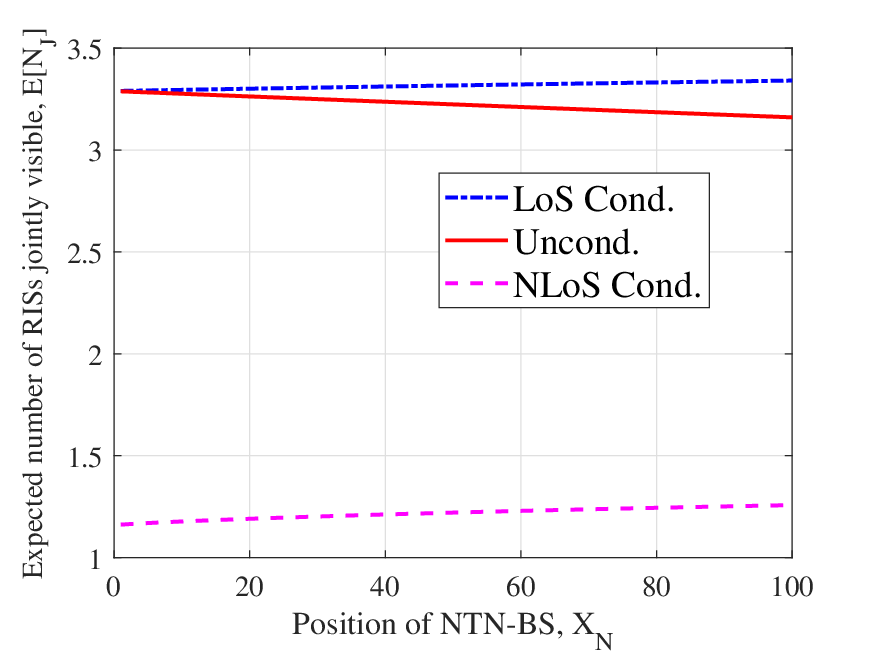}}
    \caption{Variation of Expected RIS jointly visible in a rural setting $(\lambda = 0.001)$ with NTN-BS deployed at $100$m: (a) average height of buildings $\mu$, (b) altitude of NTN-BS $H_N$, and (c) position of the NTN-BS $H_N$. }
    \label{fig:4}
\end{figure*}

\begin{figure*}
    \centering
    \subfloat[]{
    \includegraphics[width=0.3\linewidth]{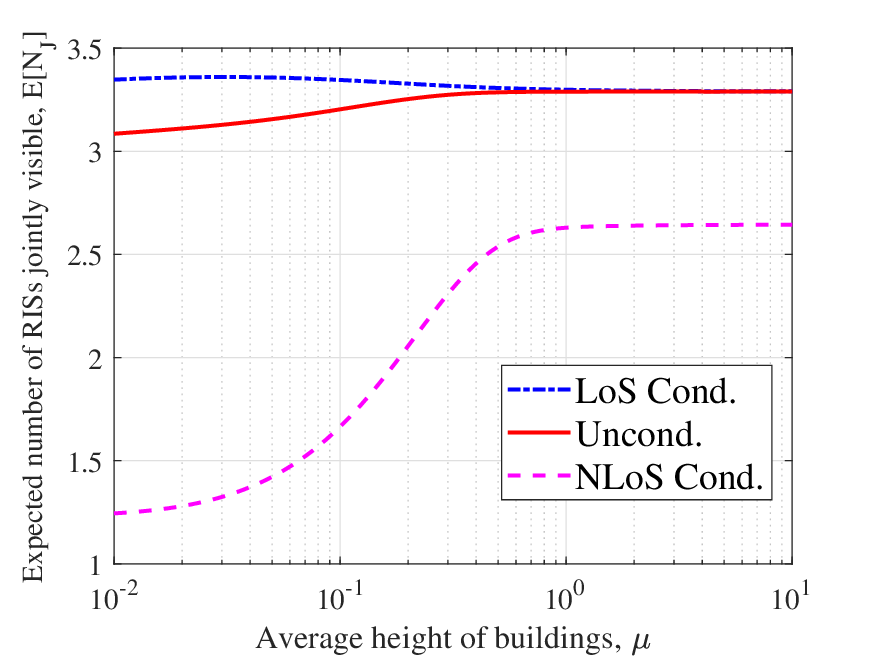}}
    \subfloat[]{
    \includegraphics[width=0.3\linewidth]{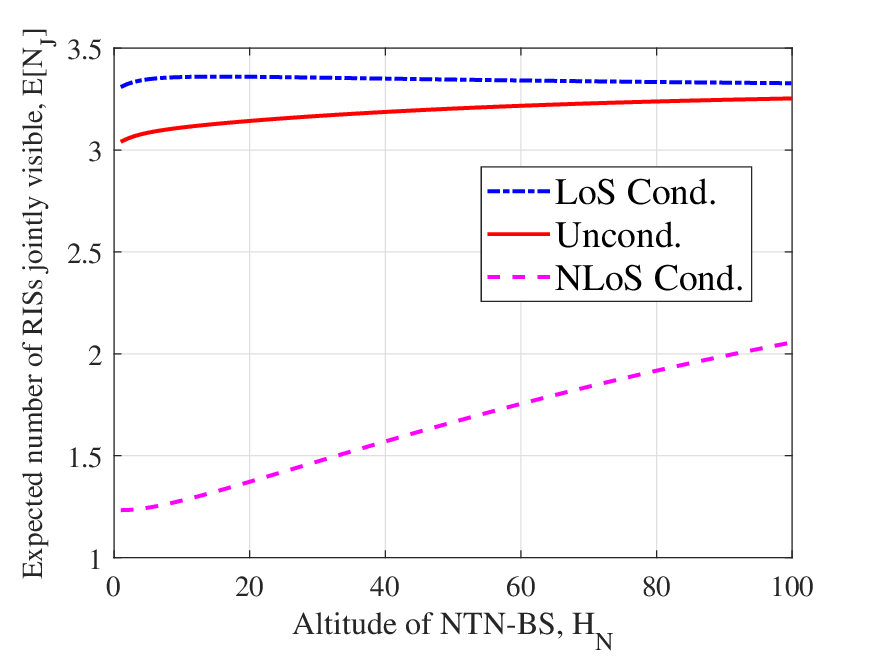}}
    \subfloat[]{
    \includegraphics[width=0.3\linewidth]{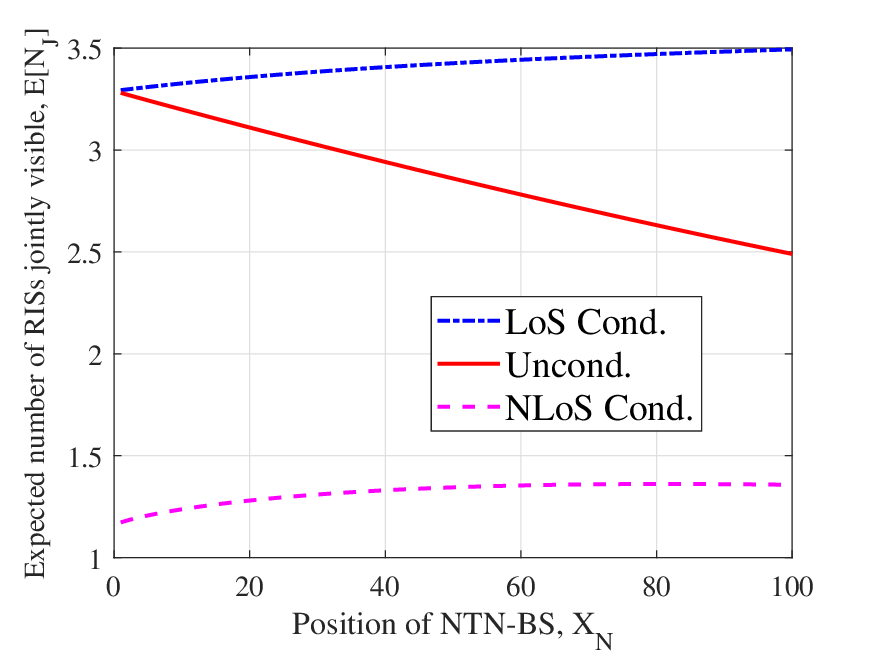}}
    \caption{Variation of Expected RIS jointly visible  in an urban setting $(\lambda = 0.007)$ with NTN-BS deployed at $100$m (a) average height of buildings $\mu$, (b) altitude of NTN-BS $H_N$, and (c) position of the NTN-BS $H_N$. }
    \label{fig:5}
\end{figure*}

\begin{figure*}
    \centering
    \subfloat[]{
    \includegraphics[width=0.3\linewidth]{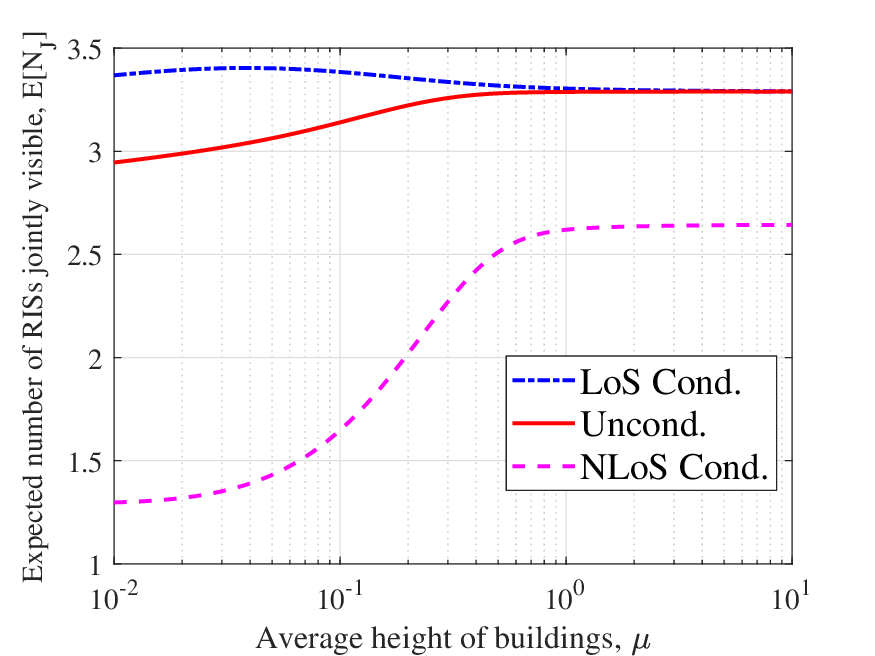}}
    \subfloat[]{
    \includegraphics[width=0.3\linewidth]{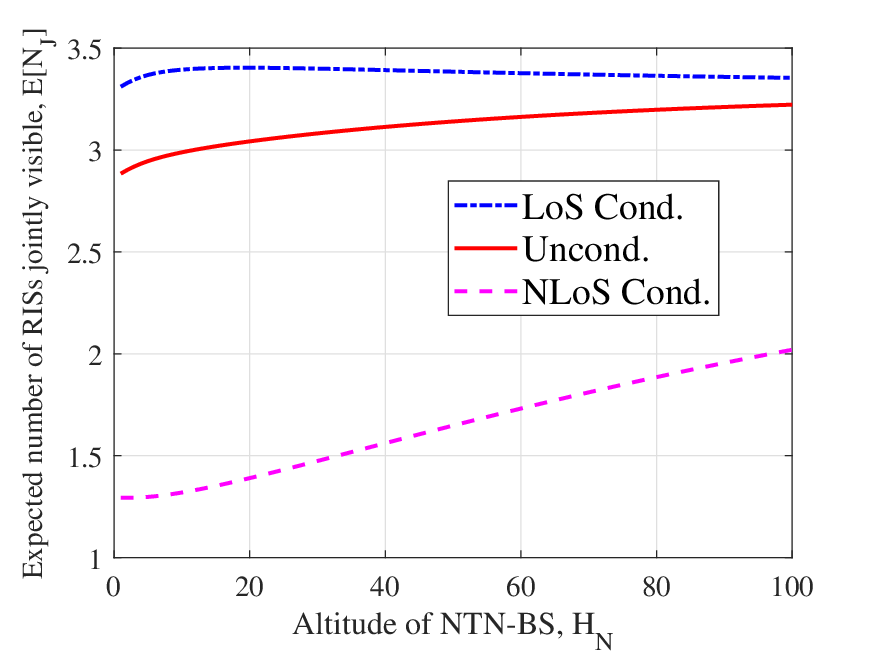}}
    \subfloat[]{
    \includegraphics[width=0.3\linewidth]{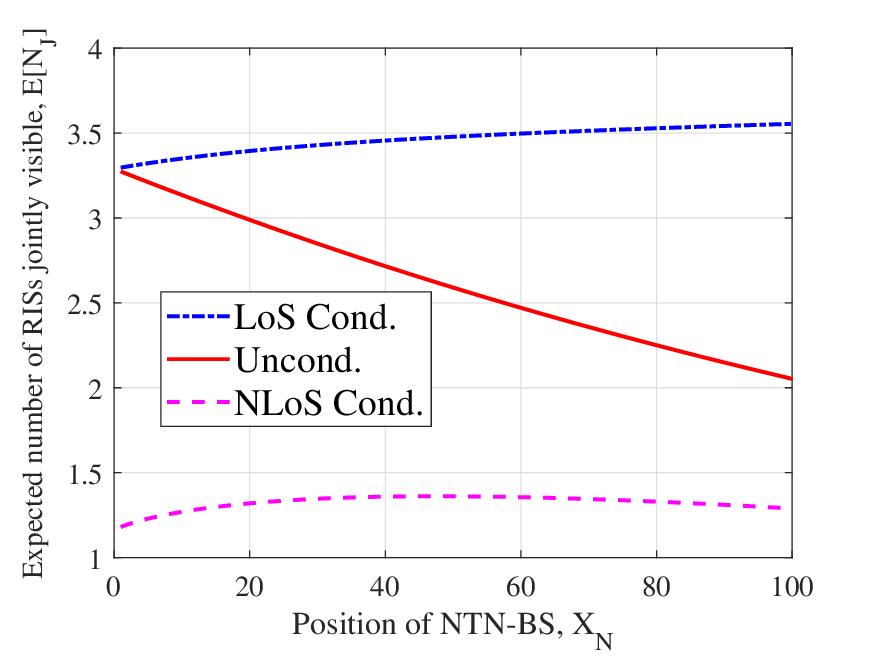}}
    \caption{Variation of Expected RIS jointly visible  in a dense urban setting $(\lambda = 0.012)$  with NTN-BS deployed at $100$m (a) average height of buildings $\mu$, (b) altitude of NTN-BS $H_N$, and (c) position of the NTN-BS $H_N$. }
    \label{fig:6}
\end{figure*}

Fig. \ref{fig:3} illustrates the variation of expected number of jointly visible RISs with the density of RISs (or buildings). It is observed that in general, the number of RISs visible conditionally to the presence of the UE-NTN LoS link, is always higher than the unconditional joint visibility or the conditional NLoS setting. The nature of expected RISs visible conditional on the UE-NTN LoS is always larger than the unconditional case as $P_{j|\mathrm{LoS}}^L \geq P_j^L,$ $P_{j|\mathrm{LoS}}^{R_1} \geq P_j^{R_1},$ $P_{j|\mathrm{LoS}}^{R_3} \geq P_j^{R_3},$ and $P_{j|\mathrm{LoS}}^{R_4} \geq P_j^{R_4}$. Only in the RHS subcase 2 is $P_{j|\mathrm{LoS}}^{R_2} \leq P_j^{R_2}$, which is due to the geometry constraint that the UE-NTN LoS does not exist at all. We believe that this RHS subcase 2 does not dominate the LHS and other RHS subcases, resulting in the expected RISs being larger when the UE-NTN link is conditionally visible, as compared to the unconditional case. Since $\mathbb{E}[N_J|\mathrm{LoS}] \geq \mathbb{E}[N_J]$, hence from  \eqref{inference-expected-joint-visibility} $\mathbb{E}[N_J|\mathrm{LoS}] \geq \mathbb{E}[N_J|\mathrm{NLoS}]$.

Further, the unconditional joint RIS probability decreases as $\lambda$ increases. This can be attributed to the fact that the UE-RIS unconditional probability $P_{ur}(x,h)$ decreases when increasing $\lambda$. This conforms with the analytical expression in Corollary \ref{cor:UR_LoS_exp}, where $P_{ur}(x,h) \propto \exp{(-\lambda)}$. Note that we infer from Theorem \ref{theo:E_Nvis_exp} that the maximum expected number of jointly visible RISs (unconditionally) is $\pi^2/3 \approx 3.289$. This maximum can be visually seen in Fig. \ref{fig:3} and in  Fig. \ref{fig:4} (refer to the red curves in these plots).

Furthermore, the expected number of RISs visible conditional on the UE-NTN NLoS is the least because, the NLoS nature of the UE-NTN link results in blocking a major region in the RHS of the UE (refer to Fig. \ref{fig:7}(c), Fig. \ref{fig:8}(c)). It is seen that as the density of RISs increases, the expected number of RISs visible in the UE-NTN NLoS conditional case decreases. This is intuitive as increasing the density of buildings increases the probability of UE-NTN NLoS, thereby decreasing the chance of expected number of RISs jointly visible.

\begin{figure*}[t]
    \centering
    \subfloat[][]{\includegraphics[width=0.3\linewidth]{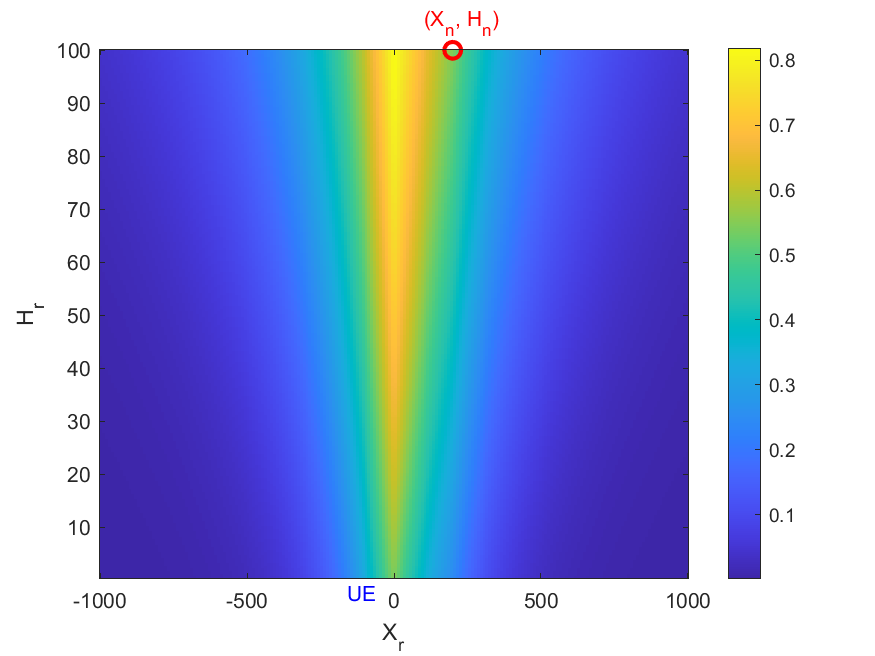}
    }
    \subfloat[][]{\includegraphics[width=0.3\linewidth]{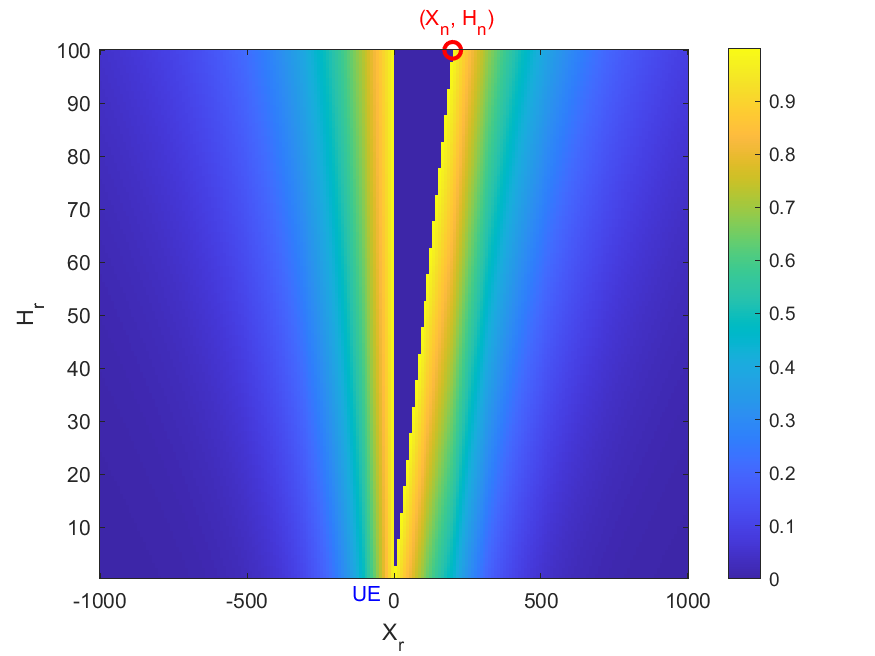}
    }
    \subfloat[][]{\includegraphics[width=0.3\linewidth]{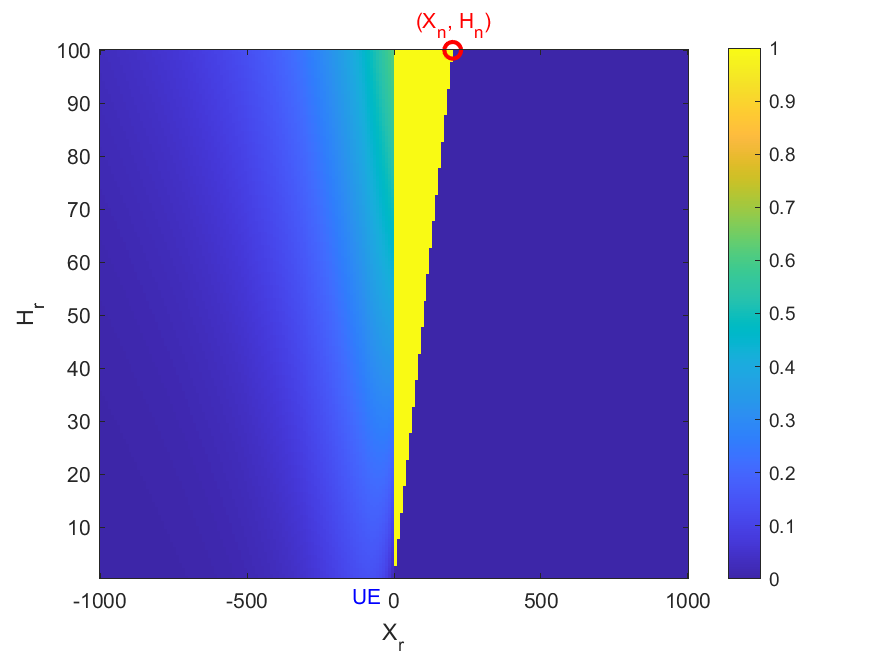}
    }
    \caption{For a UAV based BS deployed at a height of $100$m, in an urban setting $(\lambda = 0.007)$: Heatmap of probability of joint visibility of a RIS (a) unconditionally, (b) conditioned on the UE-NTN  LoS, (c) conditioned on the UE-NTN NLoS.}
    \label{fig:7}
\end{figure*}

\begin{figure*}[t]
    \centering
    \subfloat[][]{\includegraphics[width=0.3\linewidth]{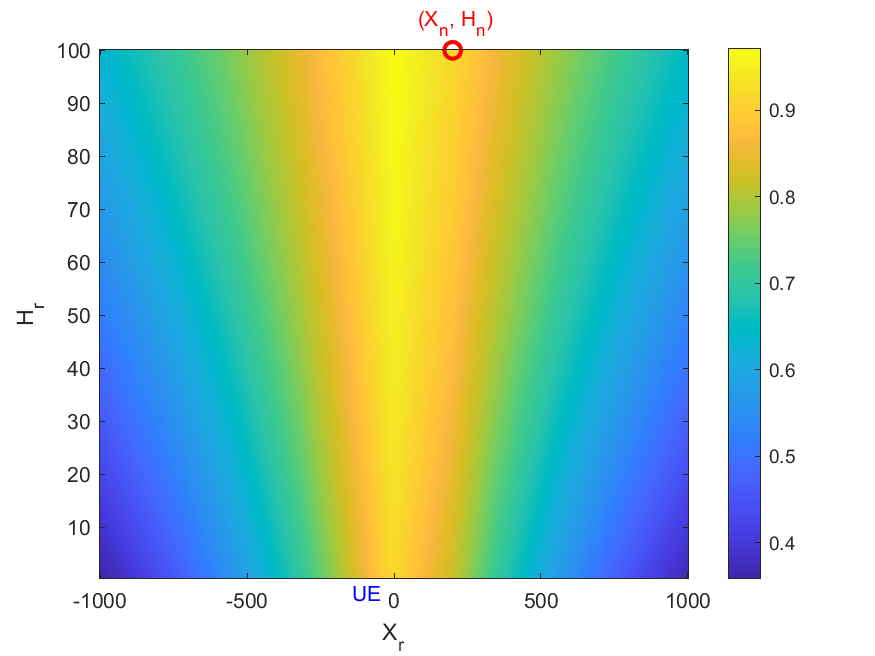}
    }
    \subfloat[][]{\includegraphics[width=0.3\linewidth]{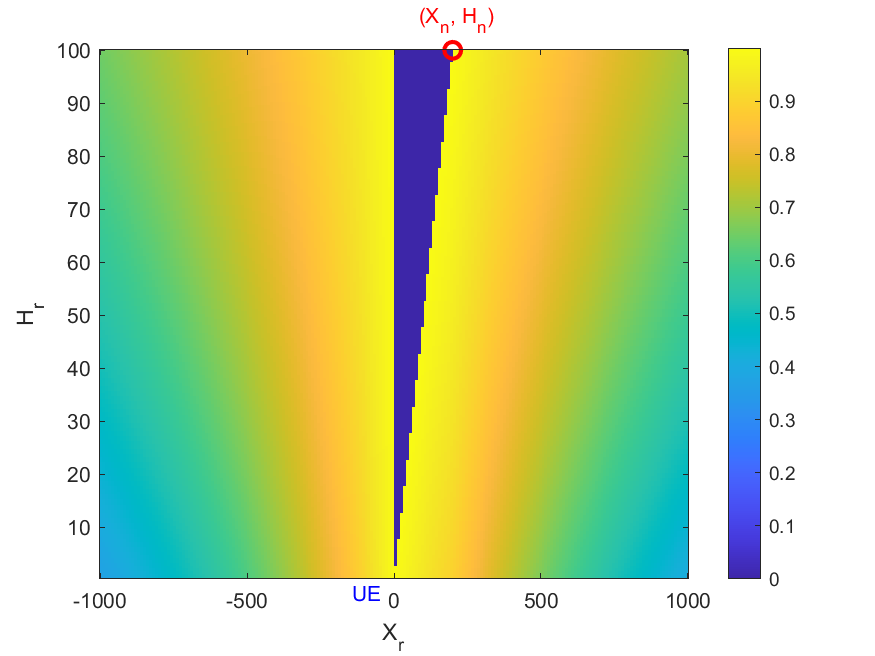}
    }
    \subfloat[][]{\includegraphics[width=0.3\linewidth]{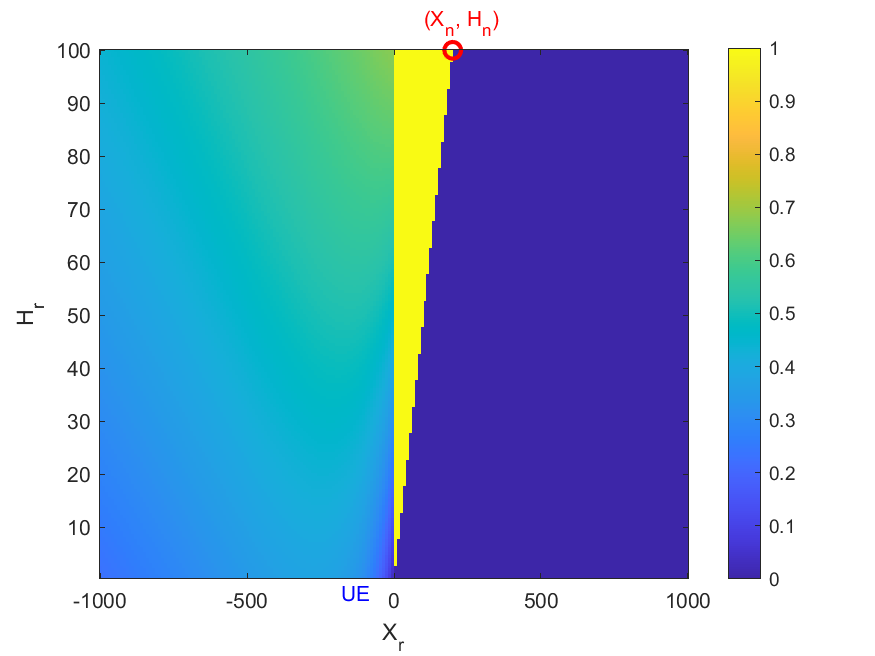}
    }
    \caption{For a UAV based BS deployed at a height of $100$m, in a rural setting $(\lambda = 0.001)$: Heatmap of probability of joint visibility of a RIS (a) unconditionally, (b) conditioned on the UE-NTN LoS, (c) conditioned on the UE-NTN NLoS.}
    \label{fig:8}
\end{figure*}

\begin{figure*}[t]
    \centering
    \subfloat[][]{\includegraphics[width=0.3\linewidth]{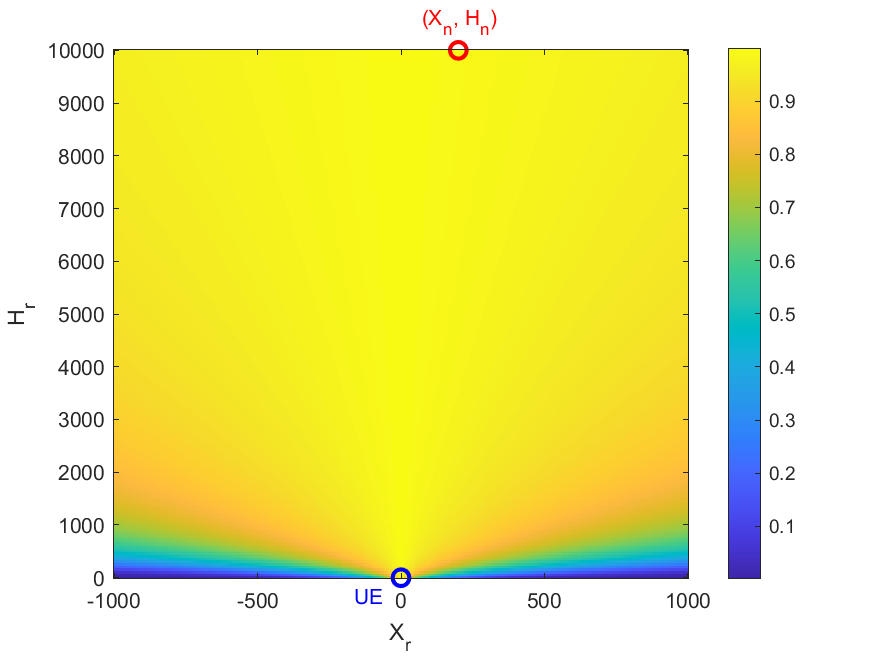}
    }
    \subfloat[][]{\includegraphics[width=0.3\linewidth]{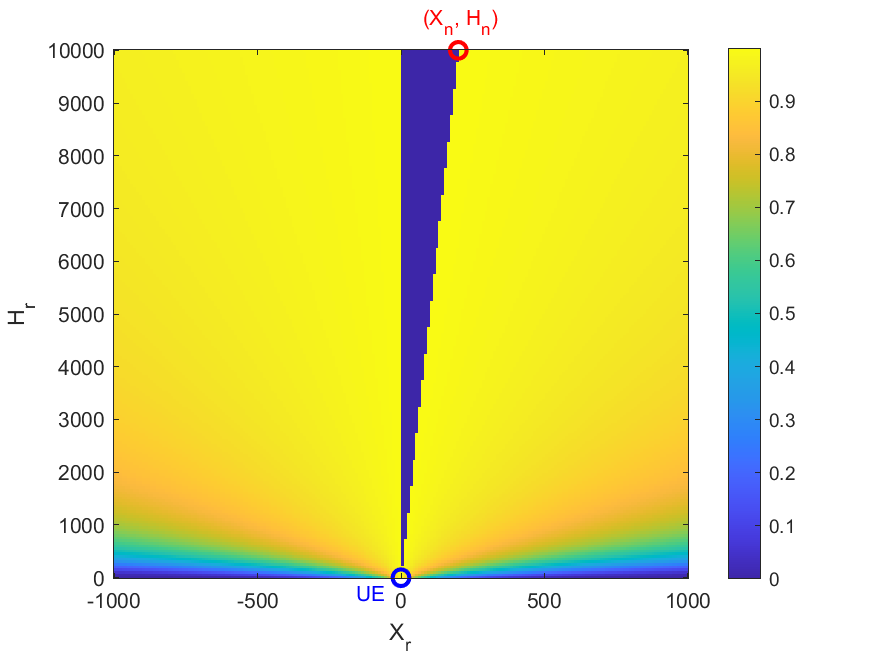}
    }
    \subfloat[][]{\includegraphics[width=0.3\linewidth]{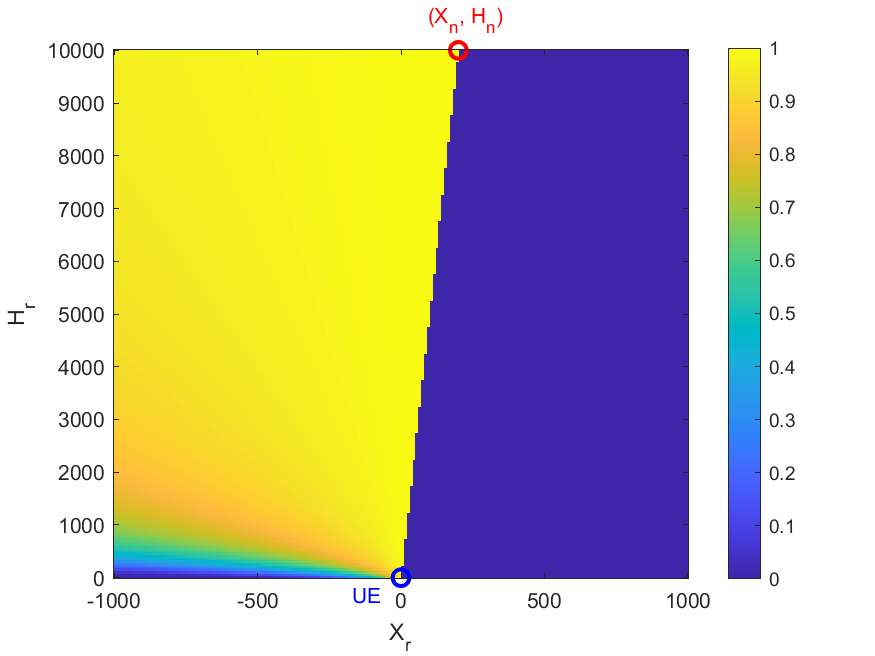}
    }
    \caption{For a high altitude NTN-BS deployed at a height of $10$ Km, in an urban setting $(\lambda = 0.007)$: Heatmap of probability of joint visibility of a RIS (a) unconditionally, (b) conditioned on the UE-NTN LoS, (c) conditioned on the UE-NTN NLoS.}
    \label{fig:9}
\end{figure*}

\subsection{Case Study: Rural/Urban/Dense Urban Scenarios}

The variations of the expected number of RISs jointly visible are plotted for three settings, namely, rural (Fig. \ref{fig:4}), urban (Fig. \ref{fig:5}), and dense-urban (Fig. \ref{fig:6}). The variation is plotted against the average inverse height of buildings $\mu$, the altitude of NTN-BS $H_n$, and the position of the NTN-BS from the UE $X_n$. 

In general, as also observed in Fig. \ref{fig:3}, the expected number of RISs visible when the UE-NTN link is conditionally visible is  consistently higher than the unconditionally, and NLoS conditional cases. Further, it is observed that the gap between the LoS conditional setting and the unconditional setting increases, as the environment becomes more dense (e.g., see the blue and red curves in Fig. \ref{fig:4}(a), Fig. \ref{fig:5}(a) and Fig. \ref{fig:6}(a)).  

The expected number of RISs conditionally on the UE-NTN LoS is observed to decrease slightly with increasing $\mu$. This is because as $\mu$ increases (blue curve in Fig. \ref{fig:4}(a), Fig. \ref{fig:5}(a), Fig. \ref{fig:6}(a)), the average height of buildings decreases, increasing the UE-NTN LoS probability (refer \eqref{p_los}, where  $P_{un}(X_n, H_n) = \mathcal{V}(0, X_n, h_{un}(\cdot))$ in Section \ref{sec:sec4}-A). As $P_{un}(X_n, H_n)$ increases, the corresponding joint RIS visibility probability conditionally on UE-NTN LoS decreases slightly (refer $P_{j|\mathrm{LoS}}^{i}, i \in \{L, R_1, R_2, R_3, R_4, R_5\}$ in Section \ref{sec:sec4}-A).

As discussed in the previous subsection, the maximum expected number of jointly visible RISs (unconditionally) is $\pi^2/3 \approx 3.289$. For the unconditional RIS visibility plots (red curves) in Figs. \ref{fig:4}--\ref{fig:6}, the curve is observed to either saturate at $\pi^2/3$ (variation with $\mu, H_n$ in Figs. \ref{fig:4}(a)-(b), \ref{fig:5}(a)-(b), \ref{fig:6}(a)-(b)) or gradually decrease from $\pi^2/3$ (variation with $X_n$ in Fig. \ref{fig:4}(c), \ref{fig:5}(c), \ref{fig:6}(c)). This effect can be attributed to the fact that increasing $\mu$ or $H_n$ result in improving the chance of joint visibility of the RIS. On the contrary, increasing $X_n$ results in the NTN-BS going further away from the UE, resulting in decreasing the odds of a RIS being jointly visible. 

In particular, the expected number of RISs conditional on the UE-NTN NLoS increases with $\mu$ and $H_n$. This is because as $\mu$ increases (Fig. \ref{fig:4}(a), Fig. \ref{fig:5}(a), Fig. \ref{fig:6}(a)), the average height of buildings decreases, leading to decreasing the UE-NTN NLoS probability (refer \eqref{p_los}, where  $(1 - P_{un}(X_n, H_n)) = 1 - \mathcal{V}(0, X_n, h_{un}(\cdot))$ in Section \ref{sec:sec4}-B). As $(1 - P_{un}(X_n, H_n))$ decreases, the corresponding joint RIS visibility  probability conditional on UE-NTN NLoS increases (refer $P_{j|\mathrm{NLoS}}^{i}, i \in \{L, R_1, R_2, R_3, R_4, R_5\}$ in Section \ref{sec:sec4}-B). 
Similarly, as $H_n$ increases (Fig. \ref{fig:4}(b), Fig. \ref{fig:5}(b), Fig. \ref{fig:6}(b)), the UE-NTN NLoS probability decreases, leading to increasing the joint RIS visibility probability conditional on the NLoS.
Furthermore, increasing $X_n$ (Fig. \ref{fig:4}(c), Fig. \ref{fig:5}(c), Fig. \ref{fig:6}(c)) is not seen to have a pronounced effect on the joint RIS visibility conditional on UE-NTN NLoS.

\subsection{Heatmap Analysis}

In Figs. \ref{fig:7}, \ref{fig:8}, and \ref{fig:9}, we show the probability heatmaps for different environment settings and different NTN altitudes. These heatmap correspond to the probability density of the joint RIS visibility, with the yellow region corresponding to high chance of a jointly visible RIS and the dark blue region corresponding to almost negligible chance of a joint visible RIS. These heatmaps are derived for the unconditional, UE-NTN LoS conditional, and UE-NTN NLoS conditional RIS visibility probabilities discussed in Section \ref{sec:sec3} and \ref{sec:sec4}.

Figs. \ref{fig:7} and \ref{fig:8} depict probability density heatmaps in an urban and rural setting, respectively, for a UAV deployed at a height of 100m. Fig. \ref{fig:9} illustrates the heatmaps for a high altitude NTN-BS deployed at  10 Km, in an urban setting.

For the unconditional RIS visibility heatmap (Figs.  \ref{fig:7}(a), \ref{fig:8}(a), \ref{fig:9}(a)), it is observed that the probabilities are symmetric on the LHS and RHS of the UE, when the UE is located at the origin and the NTN-BS at $(X_n, H_n)$ (depicted with red circle). In general, the RISs have a very good chance to be jointly visible near the UE, while the visibility probability decreases radially from the ordinate axis. Further, the environment is observed to constrict the RIS visibility, with the urban environment restricting RIS locations to be near to the UE (on either sides) while the rural environment providing a greater freedom (Compare Fig. \ref{fig:7}(a) and Fig. \ref{fig:8}(a)). Furthermore, it is inferred that as the altitude of NTN-BS increases, there is a larger chance of finding more jointly visible RISs away from the UE (Compare Fig. \ref{fig:7}(a) and Fig. \ref{fig:9}(a)).   

For the LoS conditional probability heatmaps (Fig. \ref{fig:7}(b), Fig. \ref{fig:8}(b), Fig. \ref{fig:9}(b)), the triangular region from  the ordinate axis till the slope of the UE-NTN  has zero joint RIS visibility probability. This is because,  as it is conditional on the fact that the UE-NTN link exists,  no building taller than the slope of the UE-NTN link can exist in that region. This fact is observed in all the LoS conditional plots. 

Similarly, for the NLoS conditional probability heatmaps (Fig. \ref{fig:7}(c), Fig. \ref{fig:8}(c), Fig. \ref{fig:9}(c)), the trapezoidal region beyond the slope of UE-NTN link has a zero probability of RIS visibility. This is because it is conditional on the fact that the UE-NTN link is blocked by buildings which are taller. In these heatmaps, it is observed that the best RIS deployment strategy should be on the LHS to the UE  or between the UE and NTN-BS (having higher slope than the UE-NTN link).

Note that the aim of these probability  heatmaps is to demonstrate the locations where RISs can be most useful to facilitate connectivity in NTNs, given the location of NTN-BS and the UE. The analysis can also be generalized to terrestrial networks, with a terrestrial BS replacing an NTN-BS. These heatmaps also showcase the importance of individual link visibility (e.g., UE-NTN, UE-RIS, NTN-RIS), and how the RIS deployment strategy needs to adapt to the geometric constraints.

\section{Conclusion}\label{sec:sec7}
In this work, we have probabilistically studied the prospect of RISs deployed on buildings being jointly visible from the UE and the NTN-BS.
The study accounts for the dual stochasticity arising from the building locations and heights, w.r.t. a known  NTN-BS. 
We have also studied the joint RIS visibility probability conditional on the UE-NTN link being LoS or NLoS. 
The expected number of RISs jointly visible is shown to be atmost equal to twice the Basel number. The results were used to analyze the joint RIS visibility w.r.t. the system parameters. 
Probability heatmaps are provided, which showcase the locations wherein RISs are most useful, thereby maximizing the chance of joint visibility. 
This study is expected to be useful for urban planning,  
improving upon the signal quality and economy of the cellular infrastructure. Future research directions are expected to be on quantifying the gains in signal quality from the jointly visible RISs and its effect on planning the  cellular infrastructure. 

\section*{Acknowledgments}
The work of J. Lee was supported by the National Research Foundation of Korea (NRF) grant funded by the Korea government (MSIT) (RS-2026-25473096); The work of F. Baccelli was supported by the Horizon Europe INSTINCT project (grant SNS 101139161), the France 2030 projects PEPR reseaux du Futur project (grant ANR-22-PEFT-0010),   the 5G NTN mmWave (BPIFrance) grant, and the Hubert Curien grant for collaboration with South Korea.


\appendix

\renewcommand{\theequation}{A.\arabic{equation}}
\setcounter{equation}{0}

\subsection{Proof of Lemma \ref{lem:pU}}\label{app:lemma1}
\begin{proof}
For $x{>}0$, a RIS at $(x,h)$ is visible to the UE iff,  
$\forall \ x_i {\in} (0, x)$,  
$h_i {<} \frac{h}{x} x_i$.
Hence, the UE-RIS LoS probability is 
\begin{equation*}
    P_{ur}(x,h) \triangleq  \mathbb{P}\left( \forall x_i \in \Phi \cap (0,x), h_i < \frac{h}{x}x_i \right).
\end{equation*}
Using the PGFL of PPP \cite{baccelli2009stochastic}, we have
$P_{ur}(x,h) 
= \exp\left( - \lambda \int_0^x \left[1 - F\left( \frac{h}{x}s \right)\right] ds \right).$
Substituting $u {=} \frac{s}{x}$, we get
\begin{align*}
P_{ur}(x,h) = \exp\left( - \lambda x \int_0^1 \left[1 - F(hu) \right] du \right).
\end{align*} 
For $x<0$, by symmetry of the blocking interval, $P_{ur}(x,h)=P_{ur}(|x|,h)$.
\end{proof}


\subsection{Proof of Lemma \ref{lem:pN}}\label{app:lemma2}
\begin{proof}
 For the case $x < X_n$, a RIS at  $(x, h)$ is visible to the NTN-BS (at $(X_n, H_n)$) iff, 
\begin{equation*}
    h_i < \frac{H_n - h}{X_n - x}(x_i - x) + h, \ \forall \ x_i \in (x, X_n).
\end{equation*}
Therefore, using PGFL, the NTN-RIS LoS probability is
    $P_{nr}(x,h) {\triangleq} \exp\left(-\lambda \int_x^{X_n} 
    \left[1 {-} F\left(h {+} \frac{H_n {-} h}{X_n {-} x}(s {-} x)\right)\right] ds\right).$
Substituting $s = x + u(X_n - x)$, 
we get
\begin{equation*}
    P_{nr}(x,h) {=} \exp\left(-\lambda |X_n {-} x| 
    \int_0^1 \left[1 {-} F\left(h {+} u(H_n {-} h)\right)\right] du\right).
\end{equation*}
 For $x > X_n$, by symmetry, the blocking interval $(X_n, x)$ replaces $(x, X_n)$, and the result follows.
\end{proof}


\subsection{Proof of Lemma \ref{lemma-lhs-los}}\label{app:lemma6}

 \begin{proof}
For a RIS is located  
as in \eqref{geometry-lhs-rhs1}, we have    
\begin{align}
    &P_{j|\mathrm{LoS}}^{L}(X_r, H_r) {=} \notag \\ & \mathbb{P}( h_{ur}   \text{LoS} \in [X_r, 0], h_{nr} \text{LoS} \in [X_r, X_n]  \vert  h_{un}  \text{LoS} \in [0, X_n])  \notag
    \\ 
    &= \frac{\mathbb{P}\left( h_{ur}(x) \cap h_{nr}(x) \cap h_{un}(x) \text{LoS} \ \forall \ x {\in} [X_r, X_n] \right)}{\mathbb{P}\left(h_{un}(x) \text{LoS} \ \forall \ x {\in} [0, X_n] \right)}   \notag
    \\
    &= \frac{\splitdfrac{\mathbb{P}\left(\min\{h_{ur},  h_{nr}\} \text{LoS} \in [X_r, 0]\right)}{\mathbb{P}\left(\min\{h_{nr}, h_{un}\} \text{LoS}  \in [0, X_n]\right)}}{\mathbb{P}\left(h_{un} \text{LoS} \in [0, X_n]\right)}   \label{lhs_los_1}
    \\
     &= \frac{\mathbb{P}\left(h_{ur} \text{LoS} \in [X_r, 0]\right)  \mathbb{P}\left(h_{un} \text{LoS} \in [0, X_n]\right)}{\mathbb{P}\left(h_{un} \text{LoS} \in [0, X_n]\right)}   \label{lhs_los_2}
     \\
     &= \mathbb{P}\left(h_{ur} \text{LoS} \in [X_r, 0]\right) = \mathcal{V}(X_r, 0, h_{ur}(\cdot)).  \label{lhs_los_3}
\end{align}
We infer from the geometry in Fig. \ref{fig:system_LHS} that $h_{ur}(x) \geq h_{nr}(x) \forall \ x \in [X_r, 0]$ and  $h_{nr}(x) \geq h_{un}(x) \ \forall \ x \in [0, X_n]$. Hence, the simplification of \eqref{lhs_los_1} to \eqref{lhs_los_2} is because, the NTN-RIS link is LoS whenever the UE-NTN link is LoS (which is conditionally known to exist in this case). Therefore, conditioning on the UE-NTN link being LoS, the joint RIS visibility probability comes to be \eqref{lhs_los_3} (using  \eqref{visibility-function}).

 \end{proof}


 \subsection{Proof of Lemma \ref{prop:RHS_transmissive-los}}\label{app:lemma7}

 \begin{proof}

For a RIS  
located as in \eqref{geometry-lhs-rhs2}, conditional on  the UE-NTN LoS link, we have
    \begin{align}
     &P_{j|\mathrm{LoS}}^{R_1}(X_r, H_r)   = \notag \\&  \mathbb{P}( h_{ur}  \text{LoS} \in [0, X_r], h_{nr} \text{LoS} \in [X_r, X_n]  \vert h_{un} \text{LoS} \in [0, X_n]) \notag
    \\ 
    &= \frac{\mathbb{P}( h_{ur} \text{LoS} \cap h_{nr} \text{LoS}  \cap h_{un} \text{LoS})}{\mathbb{P}(h_{un} \text{LoS})}  \label{rhs_los_case1_0}  
    \\
    &= \frac{\mathbb{P}(h_{ur} \cap h_{un} \text{LoS} \in [0, X_r])  \mathbb{P}(h_{nr} \cap h_{un} \text{LoS} \in [X_r, X_n])}{\mathbb{P}(h_{un} \text{LoS} \in [X_r, X_n])} \label{rhs_los_case1_1}
    \\
    &{=} \frac{\mathbb{P}(h_{ur} \text{LoS} {\in} [0, X_r])  \mathbb{P}(h_{nr} \text{LoS} {\in} [X_r, X_n]) }{\mathbb{P}( h_{un} \text{LoS} {\in} [X_r, X_n])} \label{rhs_los_case1_2}
    {=} \frac{P_j^{R_1}(X_r, H_r)}{\mathcal{V}(0, X_n, h_{un}(\cdot))}. 
    \end{align}
To simplify \eqref{rhs_los_case1_1} to \eqref{rhs_los_case1_2}, from geometry (Fig. \ref{fig:system_RHS}(a)), we have $h_{ur}(x) {\leq} h_{un}(x)$ for all $x {\in} [0, X_r]$ and $h_{nr}(x) {\leq} h_{un}(x)$ for all $x {\in} [X_r, X_n]$.  
The numerator in \eqref{rhs_los_case1_2} is $P_j^{R_1}(X_r, H_r)$, derived in  Lemma \ref{prop:RHS_transmissive}, while the denominator comes from \eqref{visibility-function}.

Now, 
when a RIS is located as in \eqref{geometry-lhs-rhs3}. 
From the geometry (Fig. \ref{fig:system_RHS}(b)), since $H_r {>} \frac{H_n}{X_n}X_r {=} h_{un}(X_r)$, the RIS building itself obstructs the UE-NTN link, and hence the UE-NTN link is always in NLoS.  
Therefore, the event \eqref{rhs_los_case1_0}, i.e.,  $\{h_{ur}~\text{LoS} \cap h_{nr}~\text{LoS} \cap h_{un}~\text{LoS}\} = \emptyset$, yielding $P_{j|\mathrm{LoS}}^{R_2}(X_r, H_r) = 0$.

\end{proof}


\subsection{Proof of Lemma \ref{prop:RHS_reflective-los}}\label{app:lemma8}

\begin{proof}

For a RIS located  
as in \eqref{geometry-lhs-rhs4}, conditional on the UE-NTN LoS link (see Fig. \ref{fig:system_RHS}(c)), we have
    \begin{align}
     &P_{j|\mathrm{LoS}}^{R_3}(X_r, H_r) \notag \\ &{=} \mathbb{P}( h_{ur}  \text{LoS} {\in} [0, X_r], h_{nr} \text{LoS} {\in} [X_n, X_r] \vert \ h_{un}  \text{LoS} {\in} [0, X_n]) \notag
    \\
    &{=} \frac{\mathbb{P}(h_{ur} \cap h_{un} \text{LoS} \in [0, X_n]) \mathbb{P}(h_{ur} \cap h_{nr} \text{LoS}  \in [X_n, X_r])}{\mathbb{P}(h_{un} \text{LoS} \in [0, X_n])}  \label{rhs_los_case3_1}
    \\
    &= \frac{\mathbb{P}(h_{ur} \text{LoS} \in [0, X_n]) \mathbb{P}(h_{ur} \text{LoS} \in [X_n, X_r])}{\mathbb{P}(h_{un} \text{LoS} \in [0, X_n])}  \label{rhs_los_case3_2}
    \\
    &= \frac{\mathbb{P}(h_{ur} \text{LoS} \in [0, X_r])}{\mathbb{P}(h_{un} \text{LoS} \in [0, X_n])}  \label{rhs_los_case3_3}
    = \frac{P_j^{R_3}(X_r, H_r)}{\mathcal{V}(0, X_n, h_{un}(\cdot))}.
    \end{align}
Here, \eqref{rhs_los_case3_1} simplifies to \eqref{rhs_los_case3_2}, since $h_{ur}(x) {\leq} h_{un}(x)$ for all $x {\in} [0, X_n]$ and $h_{ur}(x) {\leq} h_{nr}(x) \forall x {\in} [X_n, X_r]$. Further, the numerator in \eqref{rhs_los_case3_3} is  $P_j^{R_3}(X_r, H_r)$ derived  in Lemma \ref{prop:RHS_reflective}.

Finally, 
when the RIS is located  
as in \eqref{geometry-lhs-rhs5} and in presence of UE-NTN LoS link (see Fig. \ref{fig:system_RHS}(d)), we have
    \begin{align}
     &P_{j|\mathrm{LoS}}^{R_4}(X_r, H_r) = \notag \\ &\mathbb{P}( h_{ur} \text{LoS} \in [0, X_r], h_{nr} \text{LoS} \in [X_n, X_r] \vert  h_{un} \text{LoS} \in [0, X_n]) \notag
    \\
    &= \frac{\mathbb{P}(h_{ur} \cap h_{un} \text{LoS}  \in [0, X_n]) \mathbb{P}(h_{ur} \cap h_{nr} \text{LoS} \in [X_n, X_r])}{\mathbb{P}(h_{un} \text{LoS} \in [0, X_n])}  \label{rhs_los_case4_1}
    \\
    &= \frac{\mathbb{P}(h_{un} \text{LoS} \in [0, X_n])  \mathbb{P}(h_{nr} \text{LoS} \in [X_n, X_r])}{\mathbb{P}(h_{un} \text{LoS} \in [0, X_n])}  \label{rhs_los_case4_2}
    \\
    &= \mathbb{P}(h_{nr} \text{LoS} \in [X_n, X_r]) 
    = \mathcal{V}(X_n, X_r, h_{nr}(\cdot)).
    \end{align}

The simplification of \eqref{rhs_los_case4_1} to \eqref{rhs_los_case4_2} is because  $h_{ur}(x) \geq h_{un}(x) \forall  x \in [0, X_n]$, and $h_{nr}(x) \leq h_{ur}(x)  \forall  x  \in [X_n, X_r]$.

\end{proof}


\subsection{Proof of Lemma \ref{lemma-lhs-nlos}}\label{app:lemma9}

\begin{proof}

For a RIS located  
as in \eqref{geometry-lhs-rhs1}, conditional on the UE-NTN NLoS link (Fig. \ref{fig:system_LHS}), we have

\begin{align}
       &P_{j|\mathrm{NLoS}}^{L}(X_r, H_r) {=} \notag \\ & \mathbb{P}( h_{ur} \text{LoS} \in [X_r, 0], h_{nr} \text{LoS} \in [X_r, X_n] \vert  h_{un} \text{NLoS} \in [0, X_n]) \notag
    \\ 
    &{=} \frac{\mathbb{P}\left( h_{ur}(x) \text{LoS} \cap h_{nr}(x) \text{LoS} \cap h_{un}(x) \text{NLoS}, x \in [X_r, X_n]\right) }{\mathbb{P}\left(h_{un}(x) \ \text{NLoS} \forall x \in [0, X_n]\right)}  \notag
    \\
    &= \frac{ \splitdfrac{\mathbb{P}\left( \min\{h_{ur}(x) \ \text{LoS} \cap h_{nr}(x) \ \text{LoS}\} \forall x \in [X_r, 0]\right)}{\times \mathbb{P}\left( h_{nr}(x) \ \text{LoS} \cap h_{un}(x) \ \text{NLoS} \forall x \in [0, X_n]\right)}}{\mathbb{P}\left(h_{un}(x) \ \text{NLoS} \forall x \in [0, X_S]\right)}. \label{lemma-9-1}
\end{align}

To solve for second term in the numerator,   $\mathbb{P}\left(h_{nr}(x) \text{LoS}\right) = \mathbb{P}\left(h_{nr}(x) \text{LoS} {\cap} h_{un}(x) \text{LoS}\right) +  \mathbb{P}\left(h_{nr}(x) \text{LoS} {\cap} h_{un}(x) \text{NLoS}\right)$.
From the geometry, since $h_{un}(x) {<}  h_{nr}(x)  \forall x {\in} [0, X_n]$,  
we get $\mathbb{P}\left(h_{nr}(x) \text{LoS} {\cap} h_{un}(x) \text{LoS}\right) =  \mathbb{P}\left( h_{un}(x) \text{LoS} \right)$. Hence, \eqref{lemma-9-1} simplifies to
\begin{align}
    &P_{j|\mathrm{NLoS}}^{L}(X_r, H_r) \notag
    \\
    &{=}\frac{\splitdfrac{\mathbb{P}\left( h_{ur} \ \text{LoS}  \in [X_r, 0]\right)}{\left(\mathbb{P}\left(h_{nr} \text{LoS} \in [0, X_n]\right) {-} \mathbb{P}\left( h_{un} \text{LoS}  \in [0, X_n]\right) \right)}}{\mathbb{P}[h_{un}(x) \ \text{NLoS} \forall x \in [0, X_n]]}  \notag
    \\
    &=\frac{P_{j|\mathrm{LoS}}^{L}(X_r, H_r) \left[\mathcal{V}(0, X_n, h_{nr}(\cdot)) - \mathcal{V}(0, X_n, h_{un}(\cdot))\right]}{1 - \mathcal{V}(0, X_n, h_{un}(\cdot))}. \notag
    \end{align}
\end{proof}


\subsection{Proof of Lemma \ref{prop:RHS_transmissive-nlos}}\label{app:lemma10}

\begin{proof}
For a RIS located 
as in \eqref{geometry-lhs-rhs2}, conditional on the UE-NTN NLoS link, we have
\begin{align}
    &P_{j|\mathrm{NLoS}}^{R_1}(X_r, H_r) = \notag \\ & \mathbb{P}( h_{ur}  \text{LoS} \in [0, X_r], h_{nr} \text{LoS} \in [X_r, X_n]  \vert  h_{un} \text{NLoS} \in [0, X_n]) \notag
    \\
    &= \frac{\mathbb{P}(h_{ur} \text{LoS} \cap h_{nr} \text{LoS} \cap h_{un} \text{NLoS})}{\mathbb{P}(h_{un} \text{NLoS})}  \notag
    \\
    &= \frac{\splitdfrac{\mathbb{P}(h_{ur}\text{LoS}\cap h_{un}\text{NLoS} \in [0, X_r])} {\times \mathbb{P}(h_{nr} \text{LoS} \cap h_{un} \text{NLoS} \in [X_r, X_n]}}{\mathbb{P}(h_{un} \text{NLoS} \in [0, X_n])} = 0.
\end{align}
Clearly, from the geometry (Fig. \ref{fig:system_RHS}(a)), since $h_{ur}(x) \leq h_{un}(x)$ for all $x \in [0, X_r]$ and $h_{nr}(x) \leq h_{un}(x)$ for all $x \in [X_r, X_n]$, the event $\{h_{ur}~\text{LoS} \cap h_{nr}~\text{LoS} \cap h_{un} \text{NLoS}\}$ is not feasible.
Therefore,
$P_{j|\mathrm{NLoS}}^{R_1}(X_r, H_r) = 0$.

For subcase 2, when the RIS is located 
as in \eqref{geometry-lhs-rhs3}, conditional on the UE-NTN NLoS (Fig. \ref{fig:system_RHS}(b)), we have 
    \begin{align}
     &P_{j|\mathrm{NLoS}}^{R_2}(X_r, H_r) = \notag \\ & \mathbb{P}( h_{ur}  \text{LoS} \in [0, X_r], h_{nr} \text{LoS} \in [X_r, X_n]  \vert  h_{un} \text{NLoS} \in [0, X_n]) \notag
     \\
    &= \frac{\splitdfrac{\mathbb{P}(h_{ur}\text{LoS}\cap h_{un}\text{NLoS} \in [0, X_r])} {\times \mathbb{P}(h_{nr} \text{LoS} \cap h_{un} \text{NLoS} \in [X_r, X_n]}}{\mathbb{P}(h_{un} \text{NLoS} \in [0, X_n])} 
    \\
    &= \frac{\mathbb{P}( h_{ur} \text{LoS} \in [0, X_r] \cap h_{nr} \text{LoS} \in [X_r, X_n])}{\mathbb{P}(h_{un} \ \text{NLoS} \in [0, X_n] )}  \label{rhs_nlos_case2_1}
    \\
     &= \frac{P_j^{R_2}(X_r, H_r)}{1 - \mathcal{V}(0, X_n, h_{un}(\cdot))}.
    \end{align}
From geometry, since $H_r {>} \frac{H_n}{X_n}X_r {=} h_{un}(X_r)$, the RIS building itself obstructs the UE-NTN link, and hence $\{h_{un}\text{NLoS}\}$ always holds. 
Further, the numerator in \eqref{rhs_nlos_case2_1} is equal to $P_j^{R_2}(X_r, H_r)$ derived in Lemma \ref{prop:RHS_transmissive}.

\end{proof}


\subsection{Proof of Lemma \ref{prop:RHS_reflective-nlos}}\label{app:lemma11}

\begin{proof}

For a RIS located  
as in \eqref{geometry-lhs-rhs4}, conditional on the UE-NTN NLoS link (Fig. \ref{fig:system_RHS}(c)), we have 
\begin{align}
    &P_{j|\mathrm{NLoS}}^{R_3}(X_r, H_r) = \notag \\ &\mathbb{P}( h_{ur} \text{LoS} \in [0, X_r], h_{nr}  \text{LoS} \in [X_n, X_r] \vert  h_{un} \text{NLoS} \in [0, X_n]) 
    \\
    &{=} \frac{\mathbb{P}(h_{ur} \text{LoS} \cap h_{un}  \text{NLoS} \in [0, X_n]) \mathbb{P}(h_{nr} \text{LoS} \in [X_n, X_r])}{\mathbb{P}(h_{un}   \text{NLoS} \in [0, X_n])}  \label{rhs_nlos_case3_1}
    {=}  0.
\end{align}
From geometry, since $h_{ur}(x) {\leq} h_{un}(x) \forall x {\in} [0, X_n]$, hence the event $\{h_{un} \text{NLoS}\}$ implies that $\{h_{ur} \text{NLoS}\}$. Therefore the event $\{h_{ur} \text{LoS} {\cap} h_{un}  \text{NLoS} {\in} [0, X_n]\}$ in \eqref{rhs_nlos_case3_1} is not feasible, resulting in $P_{j|\mathrm{NLoS}}^{R_3}(X_r, H_r) {=} 0$.

Finally, 
for a RIS located  
as in \eqref{geometry-lhs-rhs5}, conditional on the UE-NTN NLoS link (see Fig. \ref{fig:system_RHS}(d)), we have
\begin{align}
    &P_{j|\mathrm{NLoS}}^{R_4}(X_r, H_r) = \notag \\ &\mathbb{P}( h_{ur} \text{LoS} \in [0, X_r], h_{nr}  \text{LoS} \in [X_n, X_r] \vert  h_{un} \text{NLoS} \in [0, X_n]) 
    \\
    &= \frac{\mathbb{P}(h_{ur} \text{LoS} \cap h_{un}  \text{NLoS} \in [0, X_n]) \mathbb{P}(h_{nr} \text{LoS} \in [X_n, X_r])}{\mathbb{P}(h_{un}   \text{NLoS} \in [0, X_n])} \label{rhs_nlos_case4_0}
    \\
    &= \frac{\splitdfrac{[\mathbb{P}(h_{ur} \text{LoS} \in [0, X_n]) - \mathbb{P}(h_{ur} \text{LoS}, h_{un} \text{LoS} \in [0, X_n])]}{\mathbb{P}( h_{nr} \text{LoS} \in [X_n, X_r])}}{\mathbb{P}(h_{un} \text{NLoS} \in [0, X_n])}  \label{rhs_nlos_case4_1}
    \\
    &= \frac{\splitdfrac{[\mathbb{P}(h_{ur} \text{LoS} \in [0, X_n]) - \mathbb{P}(h_{un} \text{LoS} \in [0, X_n]) ]}{\mathbb{P}( h_{nr} \text{LoS} \in [X_n, X_r]) ]}}{\mathbb{P}(h_{un} \text{NLoS} \in [0, X_n])}   \label{rhs_nlos_case4_2}
    \\
    &= \frac{\left(\mathcal{V}(0, X_n, h_{ur}(\cdot)) - \mathcal{V}(0, X_n, h_{un}(\cdot))\right)\mathcal{V}(X_n, X_r, h_{nr}(\cdot))}{1 - \mathcal{V}(0, X_n, h_{un}(\cdot))}. \label{rhs_nlos_case4_3}
\end{align}

The simplification of \eqref{rhs_nlos_case4_0} to \eqref{rhs_nlos_case4_1} comes from $\mathbb{P}\left(h_{ur}\text{LoS} {\cap} h_{un} \text{NLoS} \right) {=} \mathbb{P}\left(h_{ur} \text{LoS} \right) {-} \mathbb{P}\left(h_{ur} \text{LoS} {\cap} h_{un} \text{LoS} \right)$, $\forall x \in [0, X_n].$
Further \eqref{rhs_nlos_case4_2}  follows from $h_{ur}(x) \geq h_{un}(x)$ for all $x \in [0, X_n]$, which implies $\{h_{un}~\text{LoS}\} \subseteq \{h_{ur}~\text{LoS}\}$. Finally, \eqref{rhs_nlos_case4_3} uses \eqref{visibility-function} to obtain the closed form expression.

\end{proof}

\bibliographystyle{ieeetran}
\bibliography{referenceBibs.bib}

\end{document}